\documentclass[%
superscriptaddress,
nofootinbib,
 amsmath,amssymb,
 aps,
onecolumn,
showkeys
]{revtex4-2}

\usepackage{graphicx}
\usepackage{dcolumn}
\usepackage{bm}
\usepackage{hyperref}

\usepackage[
marginparwidth=1in,
]{geometry}

\usepackage{amsthm}
\usepackage{stmaryrd}
\usepackage[export]{adjustbox}

\newtheorem{definition}{Definition}[section]
\newtheorem{theorem}{Theorem}[section]

\newtheorem{lemma}{Lemma}

\newcommand{\cG}{\mathcal{G}}
\newcommand{\cC}{\mathcal{C}}
\newcommand{\cS}{\mathcal{S}}

\newcommand{\cI}{\mathcal{I}}
\newcommand{\cT}{\mathcal{T}}

\DeclareMathOperator{\Cat}{Cat}

\usepackage{todonotes}

\begin{document}


\title{Double scaling limit for the $O(N)^3$-invariant tensor model}

\author{{\bf V. Bonzom}}\email{bonzom@lipn.univ-paris13.fr}
\affiliation{Universit\'e Sorbonne Paris Nord, LIPN, CNRS UMR 7030, F-93430 Villetaneuse, France, EU}

\author{{\bf V. Nador}}\email{victor.nador@u-bordeaux.fr}
\affiliation{LaBRI, Univ. Bordeaux, 351 cours de la Lib\'eration, 33405 Talence, France, EU}

\author{{\bf A. Tanasa}}\email{ntanasa@u-bordeaux.fr}
\affiliation{LaBRI, Univ. Bordeaux, 351 cours de la Lib\'eration, 33405 Talence, France, EU}
\affiliation{H. Hulubei Nat. Inst. Phys. Nucl. Engineering,P.O.Box MG-6, 077125 Magurele, Romania, EU}

\date{\today}

\begin{abstract}
We study the double scaling limit of the $O(N)^3$-invariant tensor model, initially introduced in Carrozza and Tanasa, {\it Lett. Math. Phys.} (2016). This model has an interacting part containing two types of quartic invariants, the tetrahedric and the pillow one. For the 2-point function, we rewrite the sum over Feynman graphs at each order in the $1/N$ expansion as a \emph{finite} sum, where the summand is a function of the generating series of melons and chains (a.k.a. ladders). The graphs which are the most singular in the continuum limit are characterized at each order in the $1/N$ expansion. This leads to a double scaling limit which picks up contributions from all orders in the $1/N$ expansion. In contrast with matrix models, but similarly to previous double scaling limits in tensor models, this double scaling limit is summable.
The tools used in order to prove our results are combinatorial, namely a thorough diagrammatic analysis of the Feynman graphs, as well as an analytic analysis of the singularities of the relevant generating series.
\end{abstract}


\maketitle

\tableofcontents

\section{Introduction}

Tensor models~\cite{Gu3, Riv1, Ta1} are a natural generalization to higher dimensions of the celebrated $2$-dimensional matrix models~\cite{DFGiZJ}. Thus, both matrix and tensor models share some common features. In particular, matrix and tensor models admit a $1/N$ expansion, where $N$ is the size of the matrix, resp. of the tensor. However, the topological aspect of this large $N$ expansion is drastically different in the tensor case with respect to the matrix case.

The parameter governing the matrix $1/N$ expansion is the genus $g$ of the underlying ribbon graphs (a.k.a. as combinatorial maps in the mathematical literature) ~\cite{Kaz,Dav,BreIt}. The matrix large $N$ limit is thus dominated by planar maps. For tensor models, the parameter governing the $1/N$ expansion is often called the \emph{degree}. The existence of the $1/N$ expansion was initially proven in \cite{Gu1,GuRi,Gu2} for the so-called colored tensor models and then extended to many families of tensor models, see for example \cite{Dartois, Dario, CarrozzaHarribey2021, CarrozzaPozsgay2018, Carrozza2018, Bonzom2016, Bonzom2018} or the recent book on combinatorial physics \cite{Ta3}. The large $N$ limit is dominated by graphs of vanishing degree, whose structure depends on the set of interactions. Here we will be interested in the $O(N)^3$-invariant model of \cite{TaCa} for which the graphs of vanishing degree are the \emph{melonic} graphs (already appearing in \cite{BoGuRiRi}). These melonic graphs are a particular family of planar graphs which are in bijection with trees (a.k.a. branched polymers in the physics literature)~\cite{BoGuRiRi, GuRy}. 

Another key feature of matrix models is the double scaling limit mechanism \cite{DoSh,BreKa,GroMi}. This limit is obtained by sending the matrix size $N$ to infinity and the coupling constant $\lambda$ to the critical value $\lambda_c$, while holding some ratio of these two quantities fixed. This ratio is chosen that maps of all genus $g\geq 0$, and not just planar ones, contribute in this limit. To achieve this, one must know that the series of maps of genus $g$ behaves like $(\lambda_c-\lambda)^{\frac{5}{2}(1-g)}$ (for pure gravity). The double scaling parameter is thus
\begin{equation}
\kappa = N (\lambda_c-\lambda)^{5/4},
\end{equation}
and the 2-point function is $\sum_{g\geq0} \kappa^{2-2g} G_g$.

The double scaling limit for tensor models is defined in a similar way. This also requires to identify the singularities of the generating series of graphs of fixed degree. Adapting the scheme decomposition of maps from \cite{ChaMa} to tensor models has proven very useful. This was done first in ~\cite{GuSch} for the colored tensor model \cite{Gu1} (whose symmetry group is $U(N)^d$, $d\geq 3$), and then in \cite{TaFu} for the multi-orientable tensor model \cite{Ta2} (whose symmetry group is $U(N)\times O(N) \times U(N)$). These results allow to identify sub-families of graphs of any fixed degree which are more singular than others in the continuum, and thus dominate in the double scaling limit \cite{GuTaYo}.

Another tensor model, whose symmetry is $O(N)^3$, was defined in \cite{TaCa}. 
This model has two types of quartic $O(N)^3-$invariant interactions, a tetrahedron one and a pillow one (the names tetrahedron and pillow coming from the graphical representations of these interactions, see Sec. \ref{sec:def} below).

Tensor models and in particular the $O(N)^3$-invariant model \cite{Witten, KleTa} have received emphasized attention recently due to their large $N$ behavior being similar to that of the Sachdev-Ye-Kitaev (SYK) model \cite{Kitaev}. An particularly appealing feature of those models is that the large $N$ limit can be calculated, due to the special structure of melonic graphs, and leads to conformal invariance in the IR. Since then, several studies of various SYK-like tensor models have been done, both from a combinatorial and a mathematical physics point of view, see for example \cite{Bonzom, Bonzom2, Bonzom3, CarrozzaPozsgay2018, Pascalie} or again the recent book \cite{Ta3}. Note that these studies, and the present one included, focus on the behavior of the Feynman amplitudes with respect to $N$, and not on the part coming from spacetime integrals.

\medskip

In this paper, we study the double scaling limit mechanism for the $O(N)^3$-invariant tensor model with both tetrahedral and pillow interaction terms. We establish the existence of the double scaling limit and compute the $2$-point function in this limit. Similarly to the double scaling limits of other tensor models previously studied, we perform a scheme decomposition {\it à la} \cite{ChaMa} to identify the families of graphs which are the most singular in the continuum at any fixed order in the $1/N$ expansion. We then obtain a double scaling limit which picks up contributions from all degrees, as in the matrix case. The 2-point function is however summable and has a tree-like behavior, unlike the matrix case.

We will explain below in Section \ref{sec:def} why we consider the 2-point function as opposed to the free energy. We also believe that the techniques we use can be extended to compute any $2r$-point functions, as was done in~\cite{GuTaYo} for the multi-orientable model. In particular, the definition of the double-scaling parameter used in the sequel is the same for all $2r$-point functions of the model and for the free energy.

Our results rely on combinatorial objects already introduced in \cite{GuSch, TaFu, BeCa}: melons, dipoles, chains, schemes. Thanks to the scheme decomposition, we can describe the singularities of the series of graphs at any fixed degree in terms of the singularities of the series of melons and chains, which are simple objects.

A variant of tensor models consists in seeing them as multi-matrix models, with $D$ matrices and a symmetry $U(D)$ or $O(D)$ \cite{BonzomCombes2013, FeVa, TaFe, Ferrari}. The double scaling limit of a $U(N)^2\times O(D)$-invariant model has been found in \cite{BeCa} using again similar techniques. In a companion paper to appear we will also give the double scaling limit of other such multi-matrix models.

\paragraph*{Strategy of the paper.} As in previous papers on the double scaling limit of tensor models \cite{GuSch, TaFu, BeCa}, the strategy is the following: 
\begin{enumerate}
\item \label{enum:Classify} Classify the graphs according to their degree.
\item \label{enum:WriteGF} Write the generating series of graphs at fixed degree in terms of known series, and identify its singularities.
\item \label{enum:DS} Describe the most singular contributions and resum them using a double scaling limit.
\end{enumerate}
 
\medskip

\paragraph*{Scheme decomposition.} Let us explain Step \ref{enum:Classify} above. There is an infinite number of graphs of fixed degree. What we are looking for is packing the infinities into well controlled graphical objects, which will be the \emph{melons} and \emph{chains} (also known, in the theoretical physics literature, as ladders). Melons can be eliminated/added on every edge without changing the degree, while chains can be extended/reduced also without changing the degree. \emph{Schemes} will then be defined as graphs without melons and with minimal chains.

This gives the following bijective result:
\begin{theorem}
Any 2-point graph can be reconstructed from a unique scheme by extending some chains and adding some melons on the edges.
\label{thm:graph-scheme}
\end{theorem}
In other words, each scheme represents a family of graphs obtained by extending some chains and adding some melons. It is therefore possible to repackage the sum of all graphs of any given degree $\omega$ as a sum over schemes of the same degree. The 2-point function thus reads
\begin{equation*}
G_\omega = \sum_{\text{Schemes $\mathcal{S}$ of degree $\omega$}} P_{\mathcal{S}}(C(M), M)
\end{equation*}
where $P_{\mathcal{S}}$ is a polynomial, $C$ the generating series of chains, and $M$ the generating series of melons. The quantity $P_{\mathcal{S}}(C(M), M)$ is the amplitude resulting from the sum over all graphs in the family of the scheme $\mathcal{S}$. The singularities of $G_\omega$ may then come from the series $C$, $M$ and from the sum over schemes of degree $\omega$ if there is an infinite number of them.

If this were the case, one would look for another operation which would encode the contribution of infinitely many schemes into a new structure. However, this turns out not to be the case. We prove the following result:
\begin{theorem}
The set of schemes of a given degree is finite in the $O(N)^3$-invariant tensor model.
\label{th:sch}
\end{theorem}
Enumerating all schemes of a given degree is still a hard combinatorial problem which has not been solved. However, in the double scaling limit we only need to identify a subset of schemes of given degree, which we will be able to obtain explicitly.

\medskip
\paragraph*{Dominant schemes and double scaling limit.} The double scaling limit consists in taking the large $N$ limit while sending the coupling constant $\lambda$ to a critical value where the contribution of some families of graphs become divergent, while maintaining a certain parameter $\kappa(N, \lambda)$ fixed. Note that this parameter $\kappa$ is defined such that the contribution of non-melonic graphs are enhanced in this limit, so that graphs of arbitrarily large order can contribute. 

Since there is a finite number of schemes of any fixed degree $\omega$, the sum over schemes in $G_\omega$ is a \emph{finite} sum. All singularities then come the series $C$ and $M$, and the double scaling limit can be obtained by finding the schemes for which $P_{\mathcal{S}}(C(M),M)$ is most singular. Those schemes are said to be \emph{dominant} in the continuum limit, because they are the most divergent when the coupling constant gets close to its critical value. As we will see in the sequel, the dominant schemes are the ones which have a maximal number of a specific type of chains, called \textit{broken} chains, and they turn out to be in bijection with binary trees.

Thus, the last result of the paper is:
\begin{theorem}
\label{resultatfinal}
The dominant schemes are in bijection with rooted binary trees and the 2-point function, in the double scaling limit writes:
\begin{equation}
G_{2}^{DS}(\mu)
= M_c(\mu) \left(1 + N^{\frac{11}{12}}\sqrt{3}\frac{t_c(\mu)^\frac{1}{4}}{\left(1+6t_c(\mu)\right)^\frac{1}{2}} \frac{1-\sqrt{1-4\kappa(\mu)}}{2\kappa(\mu)^\frac{1}{2}}\right)
\end{equation}
\end{theorem}

\medskip

\paragraph*{Organization of the paper.} 

We will compute the $2$-point function in the double scaling limit $G_2^{DS}$ of the model. This is achieved in the following three sections.
\begin{description}
\item[Feynman graphs, melons, dipoles, chains and singularity analysis] In this section, we briefly present the model and its Feynman graphs. Particular emphasis is given to melonic graphs, dipoles and chains. Their generating functions are presented. We identify their critical points. The only relevant critical points for the double scaling limit are associated to a specific type of chains called \textit{broken} chains.

\item[Finiteness of the number of schemes] In this section, we give the proof of Theorem~\ref{th:sch} above. 
This proof relies on combinatorial arguments following from the proof of the multi-orientable model given in~\cite{TaFu}. 

\item[Identification of the dominant schemes and double scaling limit] In this last section, we identify all dominant schemes of a given degree and compute the double scaling limit of the 2-point function $G_2^{DS}$ of the model.
\end{description}


\section{Definition of the model and its $1/N$-expansion}
\label{sec:def}

\subsection{Feynman graphs and their degree}

The $O(N)^3$-invariant tensor model considered here was first proposed in~\cite{TaCa} as a real variant of the complex $U(N)^3$-invariant tensor models. The models involves a tensor field with components $\phi_{abc}$, each index $a,b,c$ ranging in $\{1, \dotsc, N\}$. Interactions are required to have an $O(N)^3$ symmetry, where each copy of the $O(N)$ group acts separately on an index of the tensor:
\begin{equation}
\phi_{abc} \rightarrow \phi'_{a'b'c'} = \sum_{a,b,c=1}^N O_{a'a}^1 O_{b'b}^2 O_{c'c}^3 \phi_{abc} \qquad O^i \in O(N)
\end{equation}
An index in position $i\in\{1, 2, 3\}$ is said to be of \emph{color} $i$.

Many $O(N)^3$ tensor invariants can be constructed while respecting this symmetry \cite{AvBenDub}. As usual in tensor model literature, we consider \emph{bubble polynomials}, built as follows. Take a finite set of fields, and for each color $i=1, 2, 3$ take a pairing (a.k.a. perfect matching) between the fields. Finally, contract the indices of color $i$ following this pairing to get an invariant polynomial, i.e. 
\begin{equation}
\sum_{a_i=1}^N \phi_{\dotsb a_i \dotsb} \phi_{\dotsb a_i \dotsb}
\end{equation}
These invariants can be graphically  represented by \emph{bubbles}. A bubble $B$ is a $3-$regular properly-edge-colored graph. The correspondence between bubbles and bubble polynomials is the following: 
\begin{itemize}
\item Each vertex of the bubble represents a copy of the field $\phi$.
\item One has an edge of color $i$ between two vertices if and only if the corresponding fields are contracted along the indices of color $i$.
\end{itemize}

The action of tensor models is usually a sum over some connected bubbles\footnote{This is a natural generalization of single-trace matrix models. In matrix models, bubble polynomials reduce to products of traces of the matrix to arbitrary powers. The action usually focuses on single-trace interactions, which correspond to connected bubbles.}. We further limit ourselves to quartic interactions. This leaves one quadratic invariant bubble polynomial, which gives rise to the propagator of the model, and four invariant bubble polynomials which are the interaction terms of the model. These terms are respectively called the tetrahedral, and pillow of color $i=1,2,3$ (or quartic melonic bubble) interactions. The pillow of color $i$ is by convention the one which is disconnected when the edges of color $i$ are removed. The tetrahedral bubble is obviously invariant under color permutations, but the pillows are not and are instead swapped under color permutations. They correspond to the following invariants:
\begin{align}
I_k(\phi) &= \sum_{a, b, c} \phi_{abc}\phi_{abc} = 
\begin{array}{c}\includegraphics[scale=.27]{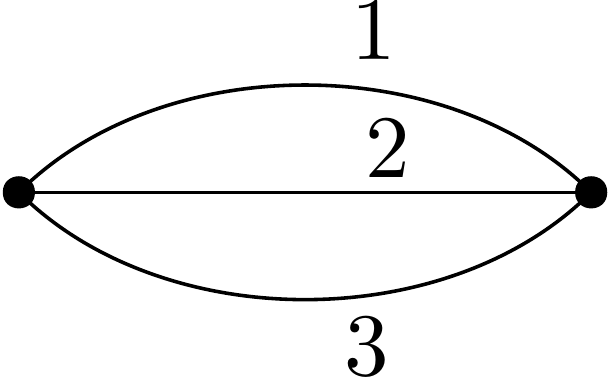}\end{array}\\
I_t(\phi) &= \sum_{a, a', b, b', c, c'} \phi_{abc}\phi_{ab'c'}\phi_{a'bc'}\phi_{a'b'c} = \begin{array}{c}\includegraphics[scale=.25]{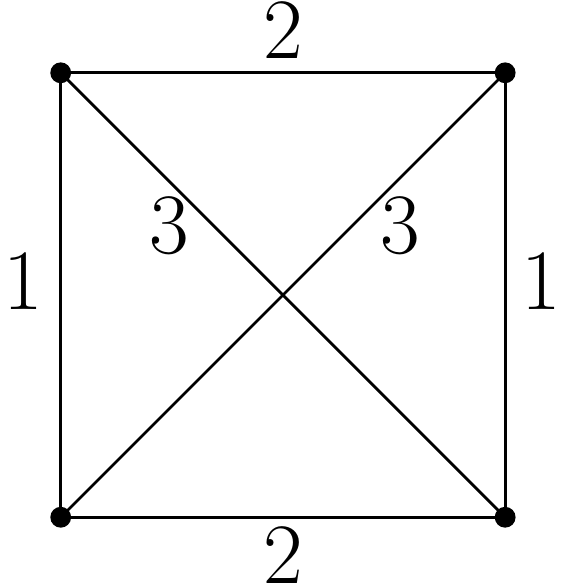}\end{array}\\
I_{p,1}(\phi) &= \sum_{a, a', b, b', c, c'} \phi_{abc}\phi_{a'bc}\ \phi_{ab'c'}\phi_{a'b'c'} = \begin{array}{c}\includegraphics[scale=.25]{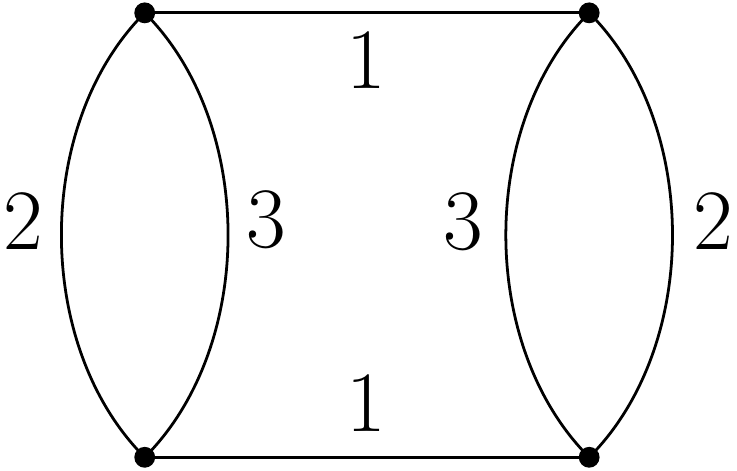}\end{array}\\
I_{p,2}(\phi) &= \sum_{a, a', b, b', c, c'} \phi_{abc}\phi_{ab'c}\ \phi_{a'bc'}\phi_{a'b'c'} = \begin{array}{c}\includegraphics[scale=.25]{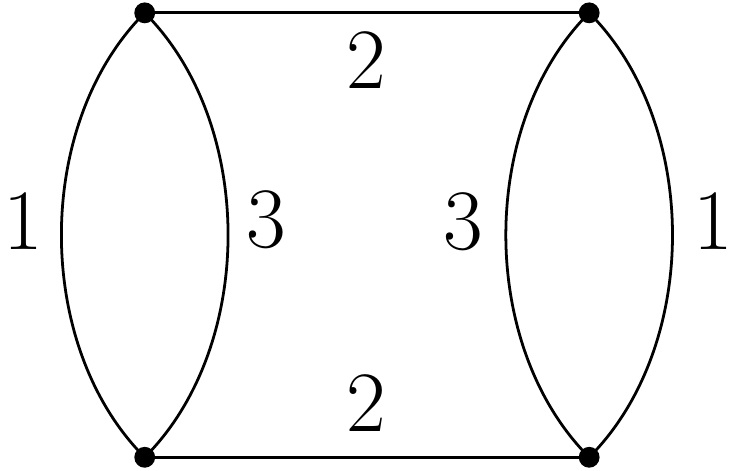}\end{array}\\
I_{p,3}(\phi) &= \sum_{a, a', b, b', c, c'} \phi_{abc}\phi_{abc'}\ \phi_{a'b'c}\phi_{a'b'c'} = \begin{array}{c}\includegraphics[scale=.25]{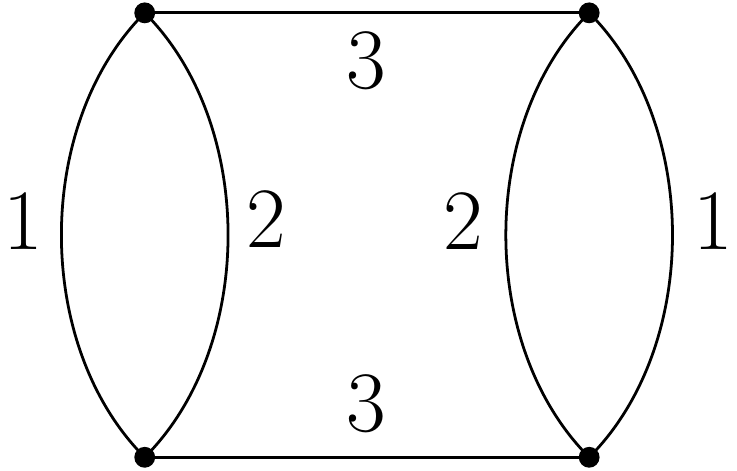}\end{array}
\end{align}

As explained above, the Feynman ``vertices'' are thus bubbles. As usual in QFT, Wick pairings connect the fields one to another. Therefore, the propagators of the Feynman graphs connect the vertices of the bubbles pairwise. A fictitious color, say $0$, can be given to the propagators, so that the Feynman graphs are $4$-regular properly-edge-colored graphs. In order not to confuse them with the edges of bubbles, we represent propagators as dashed edges. In the following, we will represent bubbles with partial coloring whenever possible i.e. when the omitted colors can be placed anyhow on the remaining pair of edges.

The action of the quartic, color-symmetric, $O(N)^3$-invariant tensor model thus writes
\begin{equation} \label{O(N)3Action}
S_N (\phi) = -\frac{N^{\frac{3}{2}}}{2}I_k(\phi) + N^{\frac{3}{2}} \frac{\lambda_1}{4}I_t(\phi) + N \frac{\lambda_2}{4}\Bigl(I_{p,1}(\phi) + I_{p, 2}(\phi) + I_{p, 3}(\phi)\Bigr)
\end{equation}
and its partition function is
\begin{equation}
Z_{N}(\lambda_1, \lambda_2) = \int \prod_{a,b,c=1}^N d\phi_{abc}\ e^{S_N(\phi)}.
\label{eq:part_fct}
\end{equation}
The scalings in $N$ of the different terms of the action, i.e. the exponents of $N$ in front of $I_k$, $I_t$ and $I_{p,i}$ for $i=1,2,3$, are chosen according to two criteria. On one hand, the free energy must admit an expansion in $\frac{1}{N}$. This gives some upper bounds on the exponents of $N$ for the interactions with respect to the scaling of the quadratic term. On the other hand, we want all interactions to contribute non-trivially to the leading order in the $\frac{1}{N}$ expansion. This gives some lower bounds on the scalings of the interactions. It has been shown in~\cite{TaCa} that for each bubble $b$ used in the action \eqref{O(N)3Action}, there is a unique scaling $\rho(b)$ so that those two criteria are satisfied. Remarkably, it admits a simple combinatorial interpretation in terms of the bubble $b$. It is given by
\begin{equation}
\rho(b) = 3 - \frac{F_b}{2}
\label{eq:scaling}
\end{equation}
where $F_b$ is the number of bicolored cycles, {\it i.e.}  the number of cycles of alternating colors $i$ and $j$ for all $i,j \in \left\lbrace 1,2,3 \right\rbrace$ with $i \neq j$. Therefore we get $\rho(k) = \rho(t) = \frac{3}{2}$ for the quadratic and tetrahedral bubbles (one bicolored cycle for every pair of colors), and $\rho(p) = 1$ for the three pillow bubbles (for example, for $I_{p,1}$ one has one bicolored cycle with colors $(1,2)$ and one with colors $(1,3)$ and two with colors $(2,3)$). In general, one can prove that if $\rho(b)$ exists, then it is unique \cite{Lionni2018, Bonzom2016}.

Thus, the $\frac{1}{N}$ expansion of the free energy reads \cite{TaCa}
\begin{equation}
F_{N}(\lambda_1, \lambda_2) = \ln Z_{N}(\lambda_1, \lambda_2) = \sum_{{ \bar{\cG}}\in{\bar{\mathbb{G}}_{}}} N^{3-\omega(\bar{\cG})} \mathcal{A}({\bar{\cG}}). 
\label{eq:largeN}
\end{equation}
The set $\bar{\mathbb{G}}$ is the set of connected $4-$regular properly-edge-colored graphs such that the subgraph obtained by removing all edges of color 0 is a disjoint union of tetrahedral bubbles and pillows and $\omega(\bar{\cG})$ is a non-negative half-integer called the degree of $\bar{\cG}$ given by
\begin{equation}
\omega (\bar{\cG}) = 3 + \frac{3}{2}n_t(\bar{\cG}) + 2n_p(\bar{\cG}) - F(\bar{\cG}),
\label{eq:deg}
\end{equation}
where
\begin{itemize}
\item $n_t(\bar{\cG})$ and $n_p(\bar{\cG})$ are respectively the number of tetrahedral and pillow bubbles in the graph,
\item $F(\bar{\cG})$ is the number of faces of the graph. Here a face is defined as a cycle of alternating colors $\{0,i\}$ for $i \in \{1,2,3\}$. The \emph{degree} of a face is the number of edges of color 0 incident to it.
\end{itemize}

\medskip

\paragraph*{2-point graphs.} The 2-point function is 
\begin{equation}
\langle \phi_{abc} \phi_{a'b'c'}\rangle = \frac{1}{N^3} G_N(\lambda_1, \lambda_2)\ \delta_{aa'} \delta_{bb'} \delta_{cc'}
\end{equation}
with $G_N(\lambda_1, \lambda_2) = \left\langle \sum_{i,j,k} \phi_{ijk} \phi_{ijk}\right\rangle$, as a consequence of the $O(N)^3$-invariance. It has an expansion on 2-point graphs, similar to that of the free energy. Denote $\mathbb{G}$ the set of 2-point graphs. From a graph $\mathcal{G}\in\mathbb{G}_{}$, we can obtain a vacuum graph $\bar{\mathcal{G}}$ by connecting the two half-edges together. $\bar{\mathcal{G}}$ is moreover equipped with a marked edge (the one formed by connecting the two half-edges), sometimes called a \emph{root}. Rooted graphs are the Feynman graphs of the expansion of $G_N(\lambda_1, \lambda_2)$.

Calculating the free energy from its Feynman expansion onto vacuum graphs requires to take into account the graph automorphisms, which can be a bit cumbersome. In contrast, rooted graphs have no symmetry. This makes the computation of the 2-point function more straightforward than that of the free energy (in fact, rooting objects is often the first step in this type of combinatorial problems).

To extract the free energy out of the 2-point function, one can introduce a coupling constant for the quadratic part of the action, then integrate the 2-point function with respect to it. This can be achieved by first rescaling the variables as follows. One rescales the coupling constants $\lambda_1, \lambda_2$ by $1/t^2$ for some parameter $t$, then rescale $\phi$ by $\sqrt{t}$ so that $t$ now only appears in front of the quadratic terms of the action. Denote $\tilde{\lambda}_{1,2} = t^2 \lambda_{1,2}$ and $\tilde{\phi} = \phi/t$, then
\begin{equation}
S_N(\phi) = -\frac{N^{\frac{3}{2}}t}{2} I_k(\tilde{\phi}) + N^{\frac{3}{2}} \frac{\tilde{\lambda}_1}{4}I_t(\tilde{\phi}) + N \frac{\tilde{\lambda}_2}{4}\Bigl(I_{p,1}(\tilde{\phi}) + I_{p, 2}(\tilde{\phi}) + I_{p, 3}(\tilde{\phi})\Bigr).
\end{equation}
In this normalization, the free energy is obtained by integrating the 2-point function with respect to $t$, at fixed $\tilde{\lambda}_1, \tilde{\lambda}_2$. The latter is
\begin{equation}
\langle \tilde{\phi}_{abc} \tilde{\phi}_{a'b'c'}\rangle = \frac{1}{t^2 N^3} G_N(\tilde{\lambda}_1/t^2, \tilde{\lambda}_2/t^2)\ \delta_{aa'} \delta_{bb'} \delta_{cc'}.
\end{equation}

\vspace{10pt}
The difference of exponents of $N$ between a 2-point graph $\cG$ and the associated vacuum graph $\bar{\cG}$ is just a $N^3$ due to opening/closing 3 faces in a systematic way. Therefore, we use as convention for the degree of $\cG\in\mathbb{G}$ the degree $\omega(\bar{\cG})$ of the graph obtained by connecting the two half-edges.





\subsection{Melons, dipoles, chains and schemes}

In this section, we introduce the different families of subgraphs which play a role in our analysis.

\paragraph{Melons\\}

Melons of the quartic $O(N)^3$-invariant tensor model were introduced and studied in~\cite{TaCa}. We give here a short summary of the relevant results.

\begin{definition}
Melonic graphs, or melons, are the graphs of vanishing degree.
\label{def:melons}
\end{definition}
The structure of melonic graphs relies on the following elementary building block:
\begin{definition}
A melonic graph is elementary if it is a 2-point graph with no non-trivial melonic $2$-point subgraph. We also say that it is an elementary melon.
\label{def:elem_mel}
\end{definition}
Another possible definition is the following: an elementary melonic graph is a 2-point, melonic graph which is 1- and 2-particle-irreducible. Melonic graphs can then be proved to be the set of graphs obtained by recursively inserting elementary melons on arbitrary edges, starting from the elementary melon itself. 



\medskip

There is an elementary melon associated to each of the quartic bubbles. An elementary melon is said to be of type I if it comes from a tetrahedral bubble and of type II if it comes from a pillow bubble.
\begin{equation}
\text{Type I} = \includegraphics[scale=.5, valign=c]{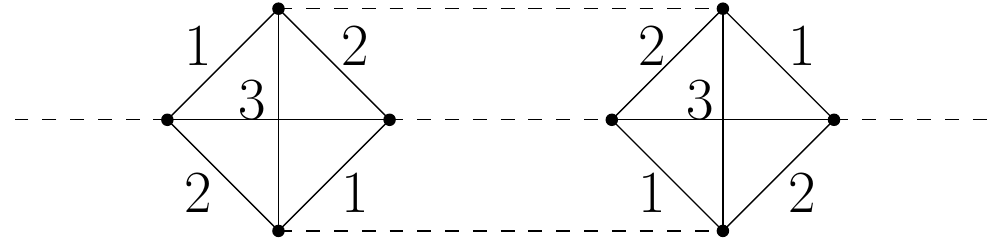} \qquad \text{Type II} = \includegraphics[scale=.5, valign=c]{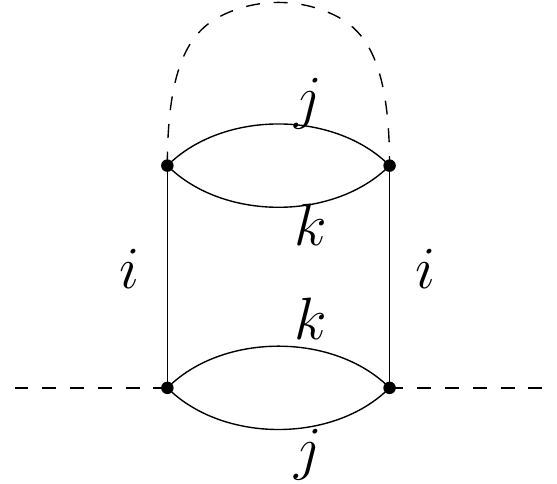}
\end{equation}
where $(a,b,c)$ is a cyclic permutation of $(1,2,3)$. Denote $M(\lambda_1, \lambda_2)$ the generating series of melonic 2-point graphs. From the recursive decomposition of melonic graphs, one has 
\begin{equation}
\includegraphics[scale=.5, valign=c]{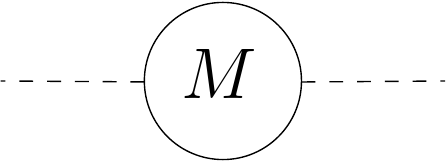} = \includegraphics[scale=.5,valign=c]{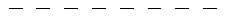} + \includegraphics[scale=.5, valign=c]{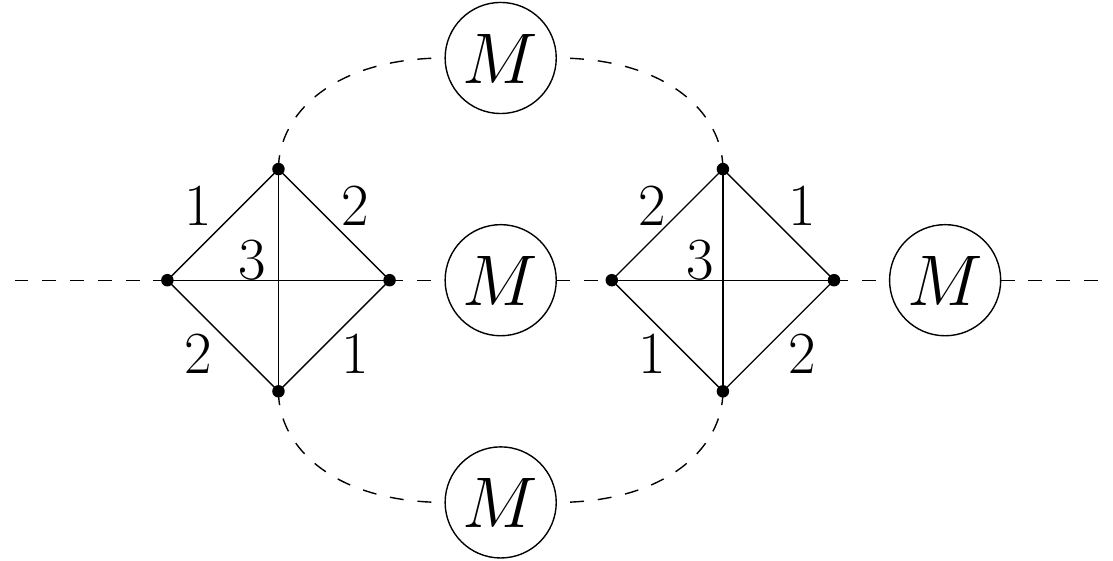} + \sum_{(i,j,k)} \includegraphics[scale=.5, valign=c]{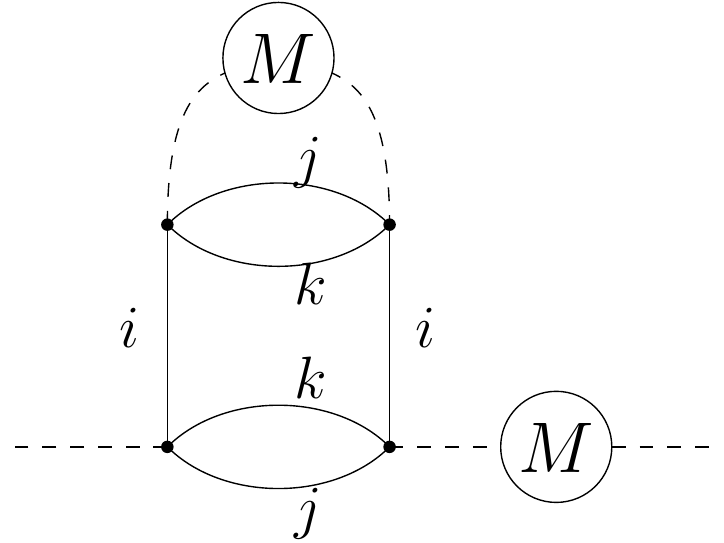}
\end{equation}
which leads to the following equation
\begin{equation}
M(\lambda_1,\lambda_2) = 1 + \lambda_1^2M(\lambda_1,\lambda_2)^4 + 3\lambda_2M(\lambda_1,\lambda_2)^2
\end{equation}
We change to the variables $(t,\mu) = (\lambda_1^2, \frac{3\lambda_2}{\lambda_1^2})$ (while retaining the notation $M$ for the melonic 2-point function), so that
\begin{equation}
M(t,\mu) = 1 + tM(t,\mu)^4 + t\mu M(t,\mu)^2
\label{eq:mel}
\end{equation}
which is the equation satisfied by the generating function of uncolored melons studied in~\cite{TaCa}.

In particular, its critical points are known. For a fixed value of $\mu \geq 0$, there is a single critical value $t_c(\mu)$ such that $(t_c(\mu),\mu)$ is a critical point of $M(t,\mu)$ and its behaviour near this critical point is
\begin{equation}
M(t,\mu) \underset{t \rightarrow t_c(\mu)}{\sim} M_c(\mu) + K(\mu)\sqrt{1 - \frac{t}{t_c(\mu)}}
\label{eq:crit_behav}
\end{equation}
where $M_c(\mu)$ is the unique positive real root of the polynomial equation 
\begin{equation} \label{eq:Mc_poly}
-3x^3+4x^2-\mu x +2\mu = 0,
\end{equation}
and $K(\mu) = \sqrt{\frac{M_c(\mu)^2\left(M_c(\mu)^2 + \mu \right)}{6M_c(\mu)^2+\mu}}$.

\paragraph{Dipoles\\}

We defined dipoles as follows:

\begin{definition}
A dipole is a $4$-point graph obtained by cutting an edge in an elementary melon.
\label{def:dip}
\end{definition}
Starting from an elementary melon of type I, cutting two edges of color 0 leaves a single face of degree $2$ untouched. If this face has color $i$, we get a dipole of type I and color $i$,
\begin{equation}
\includegraphics[scale=.45,valign=c]{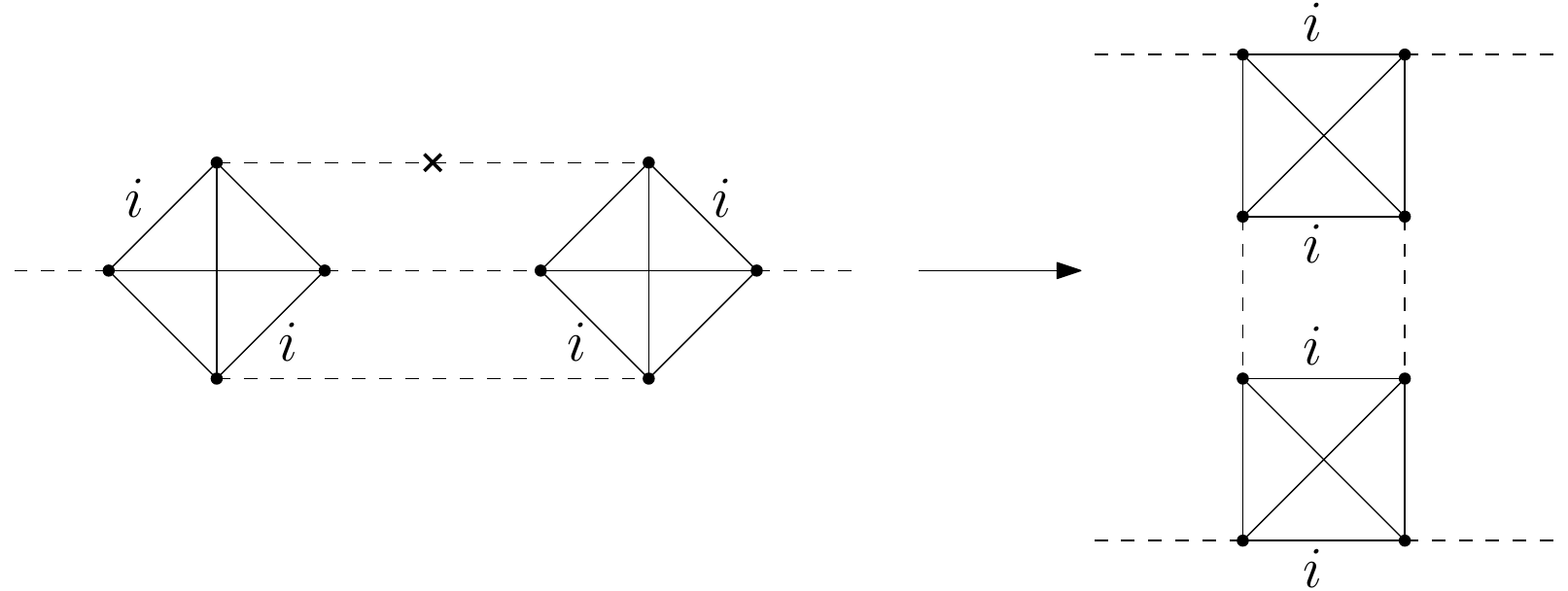}
\end{equation}
A dipole of type II is obtained by cutting the edge of color 0 of an elementary melon of type II. Dipoles of type II are simply given the color of their corresponding pillow interaction, 
\begin{equation}
\includegraphics[scale=.45,valign=c]{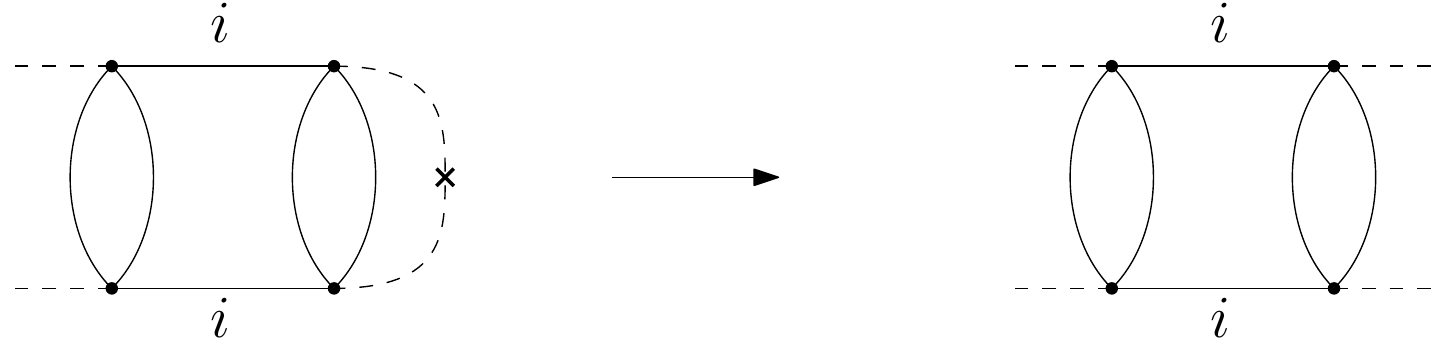}
\end{equation}
Note that dipoles of type II are always bubble-disjoint, meaning that two of them can never share a bubble. However dipoles of type I may not be bubble-disjoint, in which case we say that they are \emph{non-isolated}. It turns out that non-isolated dipoles can only occur in the following subgraph,
\begin{equation}
    \includegraphics[scale=.45]{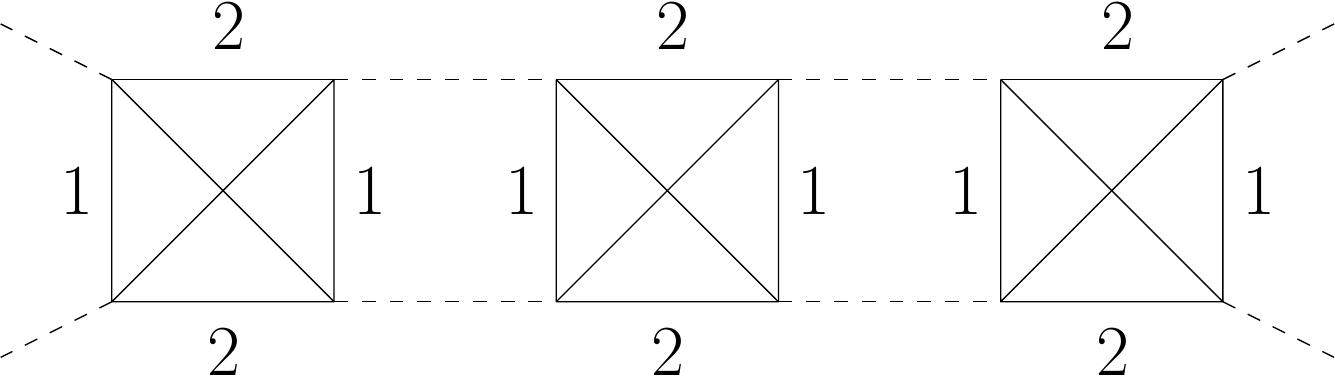}
\end{equation}
up to color permutations. Dipoles which are bubble-disjoint from others are said to be \emph{isolated}.

It is easy to check that changing a dipole of type I and color $i$ in a graph $\cG$ for a dipole of type II of the same color does not change the degree of $\cG$. Therefore, we can replace dipoles of either type which are bubble-disjoint by a \emph{dipole-vertex} which represents any of them. In terms of generating series, it means that the dipole-vertex of color $i$ is the sum of the two dipoles of color $i$,
\begin{equation}
\includegraphics[scale=0.45,valign=c]{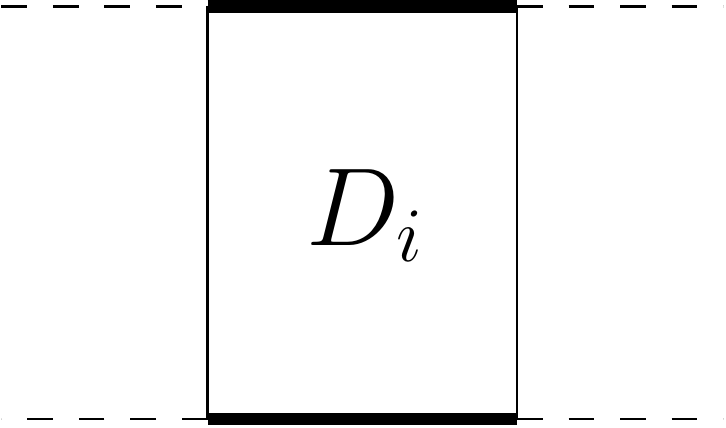} = \includegraphics[scale=0.5,valign=c]{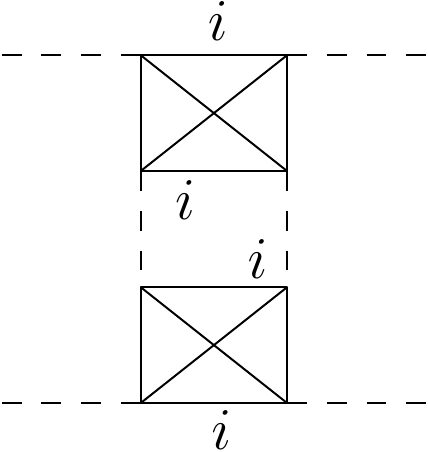} + \includegraphics[scale=0.5,valign=c]{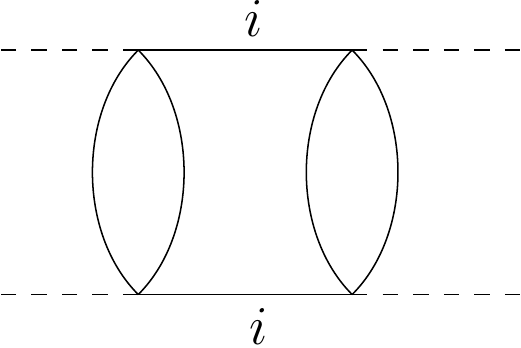}
\label{eq:dip}
\end{equation}
We have added fat edges on two opposite sides of the dipole-vertex in order to have the same symmetry as in the dipoles of type I and II. We can thus remember which external legs come from cutting the same edge in the original elementary melon: they sit on the same side of the fat edges.

Since all colors play the same role in this model, we only have to introduce one generating function for the three of them. Inserting generating functions of melons on one side of the dipoles, the generating function for dipoles reads
\begin{align}
U(t,\mu) &= \bigl(\lambda_1^2 M(t,\mu)^2 + \lambda_2\bigr) M(t,\mu)^2 \\
 &=tM(t,\mu)^4 + \frac{1}{3}t\mu M(t,\mu)^2 \underset{\eqref{eq:mel}}{=} M(t,\mu)- \frac{2}{3}t \mu M(t,\mu)^2 - 1
\label{eq:dip_rec}
\end{align}
Thus, the critical points of $U(t,\mu)$ are the same as the critical points of $M(t,\mu)$.

\vspace{5pt}
\paragraph{Chains}
\begin{definition}
A chain is either an isolated dipole, or a $4$-point function obtained by connecting an arbitrary number of dipoles by matching one side of a dipole to another side of a distinct dipole.
\label{def:chains}
\end{definition}

The \emph{length} of a chain is defined as the number of dipoles that composes the chain. By convention a chain can have just a single dipole provided it is bubble-disjoint from other dipoles. Crucially, changing the length of a chain leaves the degree unchanged. Notice that a chain of length $\ell$ contains subchains of all lengths $2\leq \ell'\leq\ell$. A chain is said to be \emph{maximal} in a graph $\bar{G}$ if it cannot be included in a longer chain in $\bar{G}$. Two different maximal chains are necessarily bubble-disjoint.
A chain is said to be of color $i$ if it only involves dipoles of color $i$,
\begin{equation}
\includegraphics[scale=0.35,valign=c]{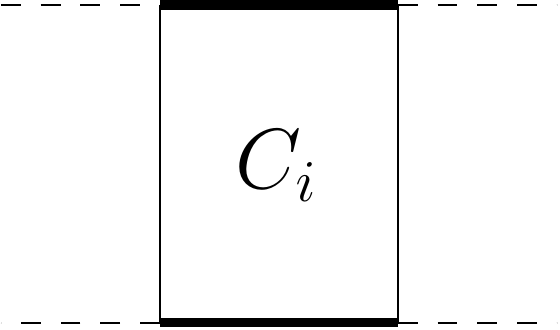} = \includegraphics[scale=0.35,valign=c]{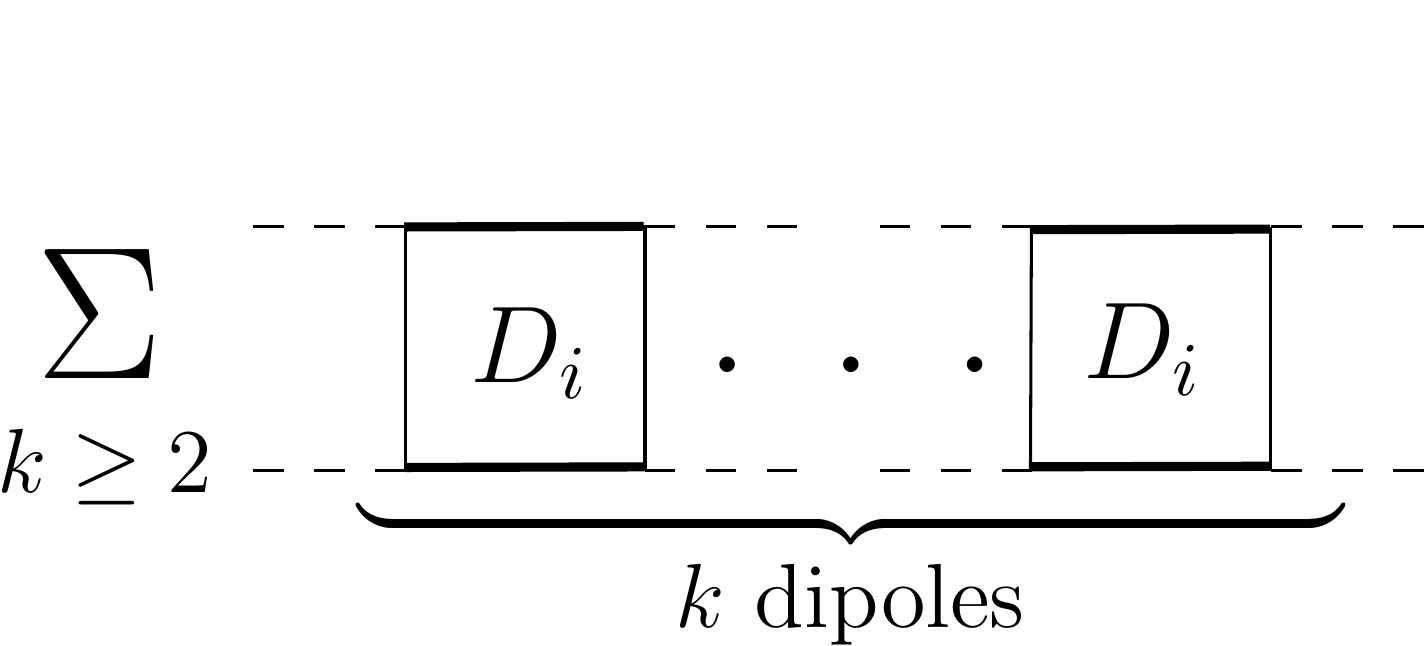}
\label{eq:chain}
\end{equation}
Otherwise it is said to be \textit{broken}. Here we represent a chain of arbitrary length by a \emph{chain-vertex} (on the left hand side). In the case of broken chains, we label the chain-vertex with $B$ instead of $C_i$.

In the following, chains of color $i$ will be denoted by a chain-vertex labelled $C_i$ and broken chains by a chain-vertex labelled $B$. A chain of color $i$ is a sequence of dipoles of length at least $2$, with melons inserted between the dipoles. Those melons have already been inserted on exactly one side of the generating series of the dipole-vertex $U(t,\mu)$. Therefore we have:
\begin{equation}
C_i(t,\mu) = U(t,\mu)^2 \sum_{k\geq 0} U(t,\mu)^k = \frac{U(t,\mu)^2}{1-U(t,\mu)}
\end{equation}
Broken chains are obtained whenever a sequence of dipoles of any color does not lead to a colored chain. Therefore their generating function is that of all chains minus those of chains of colors $i=1,2,3$:
\begin{equation}
\begin{aligned}
B(t,\mu) &= \bigl(3U(t,\mu)\bigr)^2 \sum_{k\geq 0} \bigl(3U(t,\mu)\bigr)^k -\sum_{i=1}^{3} C_i(t, \mu) \\
&= \frac{\bigl(3U(t,\mu)\bigr)^2}{1-3U(t,\mu)} - 3 \frac{U(t,\mu)^2}{1-U(t,\mu)} = \frac{6U(t,\mu)^2}{(1-3U(t,\mu))(1-U(t,\mu))}
\end{aligned}
\end{equation}
Critical points for chains are either critical points for $M(t,\mu)$, or points where $U(t,\mu) = 1$ or $U(t,\mu) = 1/3$.




\section{Schemes of the $O(N)^3$ model}

\begin{definition}
The scheme $\mathcal{S}$ of a 2-point graph $\mathcal{G}$ is obtained by first removing all melonic 2-point subgraphs, then replacing all maximal (broken and not broken) chains with chain-vertices (of the same type).
\label{def:scheme}
\end{definition}

The scheme of a graph does not depend on the order of melon removals and is thus uniquely defined. Conversely, a graph $\mathcal{G}$ is uniquely obtained by taking its scheme, first re-extending the chain-vertices as chains and then adding melons. This is precisely the content of Theorem \ref{thm:graph-scheme}.

Note that a scheme has no pillow: every pillow in a graph $\cG$ is an isolated dipole, and is therefore part of a maximal chain, the latter becoming a chain-vertex in a scheme.

In this section, we will prove Theorem~\ref{th:sch}. The result holds thanks to the following two lemmas, which ensure that there are finitely many schemes of any degree $\omega$.
 
\begin{lemma}
A scheme of fixed degree $\omega$ has finitely many chain-vertices.
\label{lemma:fst}
\end{lemma}

\begin{lemma}
A vacuum graph $\bar{\cG}\in\bar{\mathbb{G}_{}}$ of degree $\omega$ with $k$ isolated dipoles has a bounded number of bubbles, i.e. $n(\bar{\cG})\leq n_{\omega,k}$.
\label{lemma:sec}
\end{lemma}

In the proof of Lemma \ref{lemma:fst} and of the dominant schemes in the double-scaling limit, a key role is played by the following construction. If $\cS$ is a scheme, construct its \emph{skeleton graph} $\cI(\cS)$ such that
\begin{itemize}
\item The vertex set of $\cI(\cS)$ is the set of connected components obtained by removing all chain-vertices of $\cS$.
\item There is an edge between two vertices in $\cI(\cS)$ if the two corresponding connected components are connected by a chain-vertex. This edge is labeled by the type of chain-vertex.
\end{itemize}
In other words, $\cI(\cS)$ is the incidence graph between chain-vertices, and the disjoint connected components obtained after removing those. An example of a scheme and its skeleton graph is given on Figure~\ref{fig:skeleton_graph}.

\begin{figure}
\centering
\includegraphics[scale=0.60]{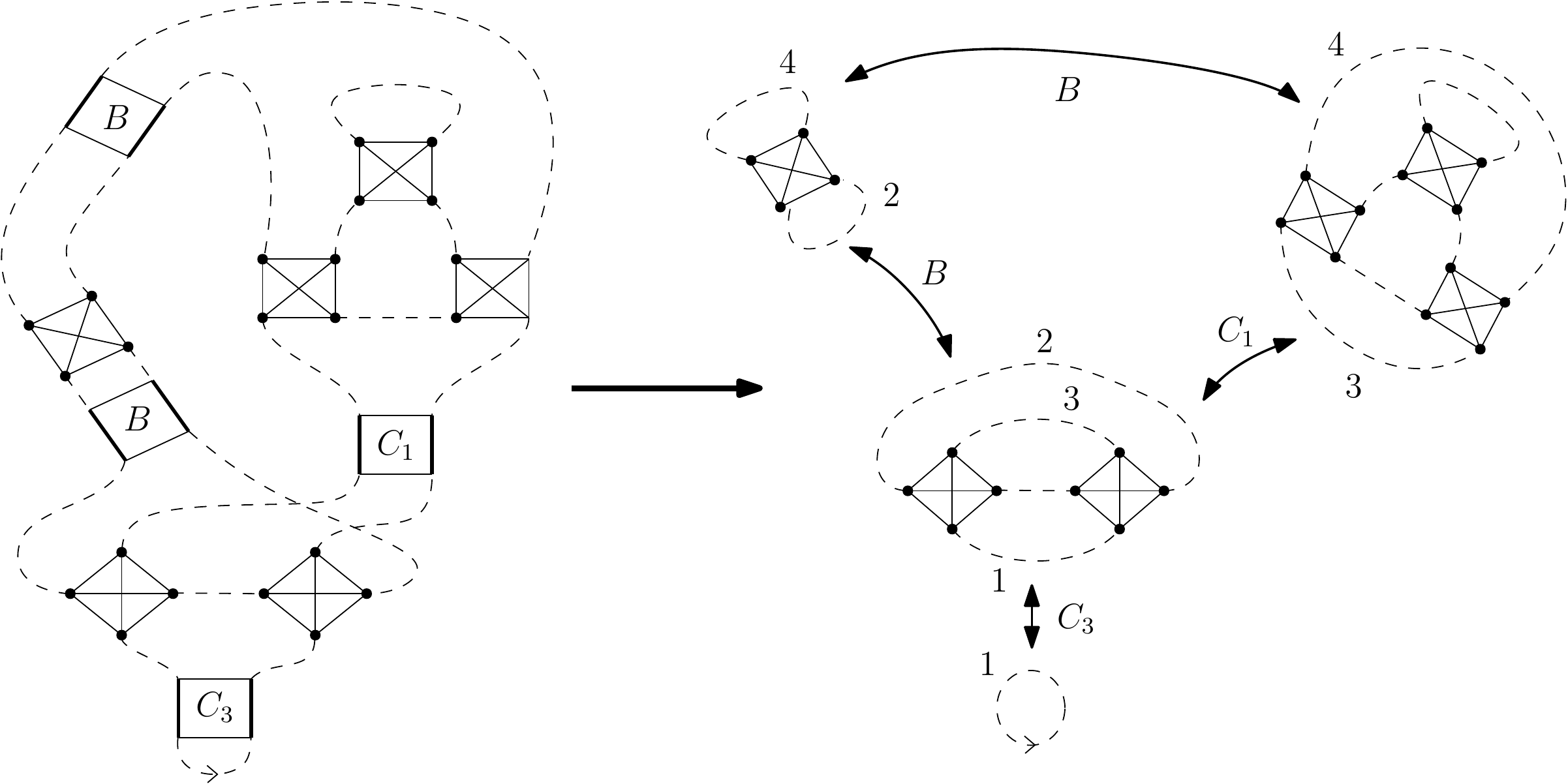}
\caption{A scheme $\cS$ and its skeleton graph $\cI(\cS)$.}
\label{fig:skeleton_graph}
\end{figure}

\subsection{Dipole and chain removals}

For the rest of the analysis, it is important that those connected components obtained by removing the chain-vertices of a scheme can be seen as graphs from $\mathbb{G}$ or $\bar{\mathbb{G}}$ and associated a degree. We thus define the removal of chain-vertices formally.

\begin{definition}
A \emph{dipole removal} in $\cG$ consists in removing the dipole and re-connecting the half-edges which were connected on the same side of the dipole together, and similarly for maximal chain removals,
\begin{equation}
\includegraphics[scale=0.7, valign=c]{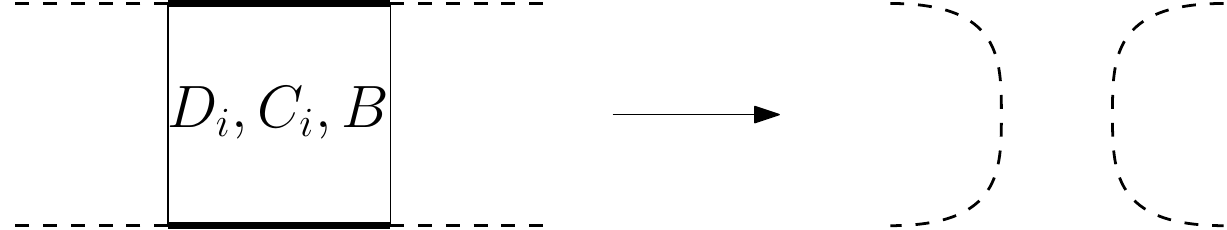}
\label{fig:dip_rem}
\end{equation}
A dipole or a chain is said to be \emph{separating} if its removal disconnects $\cG$, and non-separating if it does not. Note that a separating dipole/chain removal creates two connected components, one being a 2-point graph, the other being a vacuum graph.
\end{definition}

Note that dipole removals can be applied to non-isolated dipoles and not only dipoles which are part of chains.

In a scheme, the removal of a chain-vertex does not in general leads to another scheme, because it can create new maximal chains. Therefore, a \emph{chain-vertex removal} is defined by replacing the chain-vertex with any chain of the same type, performing the removal, then re-identifying the maximal chains in the newly obtained graph and thus finding its corresponding scheme.

\begin{lemma} \label{thm:ChainRemoval}
Let $S$ be a scheme and consider a chain-vertex removal. If it is a non-separating chain, denote $S'$ the resulting scheme, then
\begin{equation} \label{NonSeparatingRemoval}
\omega(S)-3\leq \omega(S')\leq \omega(S)-1.
\end{equation}
If it is a separating chain or dipole, denote $S_1$ and $\bar{S}_2$ the two resulting schemes, then
\begin{equation} \label{SeparatingRemoval}
\omega(S) = \omega(S_1) + \omega(\bar{S}_2).
\end{equation}
\end{lemma}

\begin{proof}
It is enough to consider the case of graphs, and the case of schemes follows directly. Moreover, among all possible types of dipoles and chains, it is enough to consider only dipoles of type I, and to consider maximal chains of length 2 made of two dipoles of type I (this is just a choice, any other works) as follows
\begin{equation}
\begin{aligned}
&\includegraphics[scale=.5, valign=c]{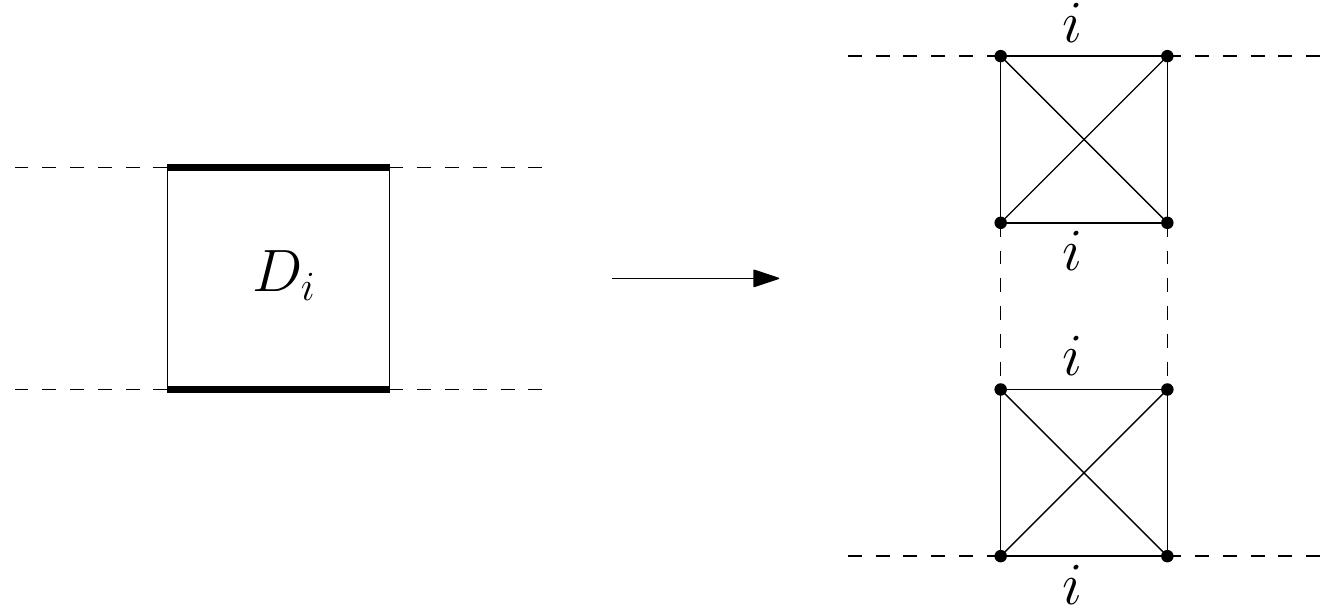}\\
&\includegraphics[scale=.5, valign=c]{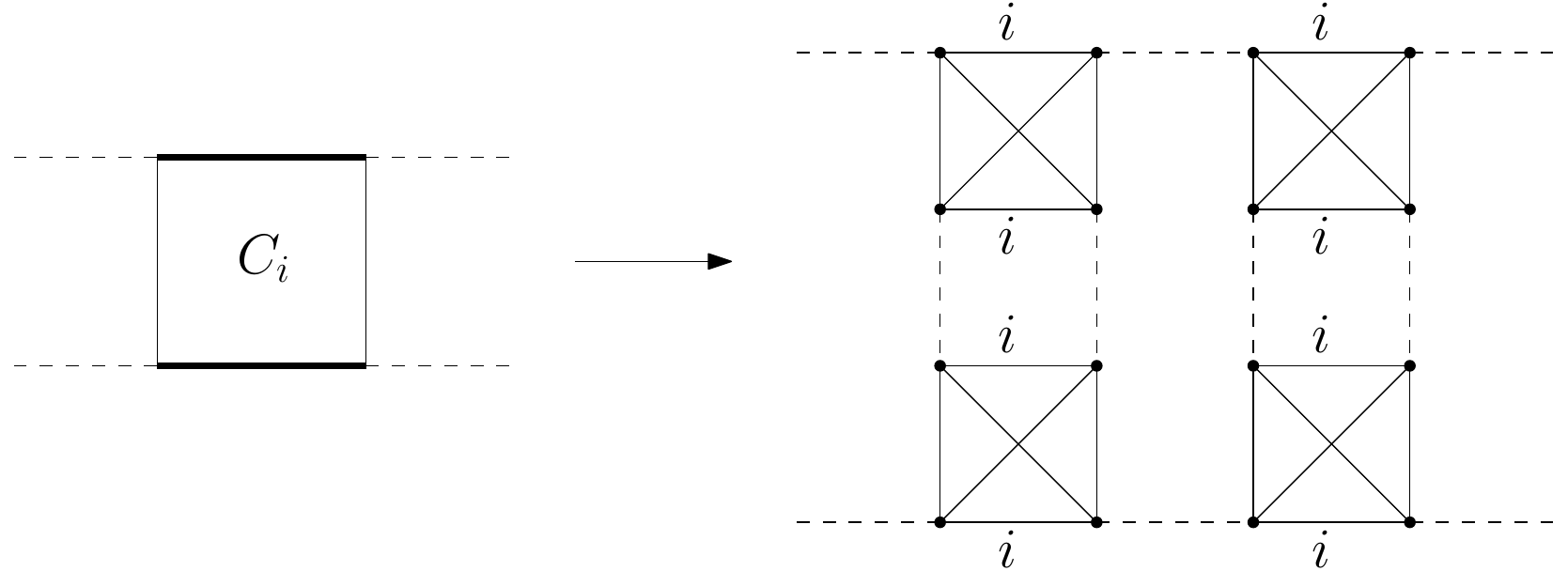}\\
&\includegraphics[scale=.5, valign=c]{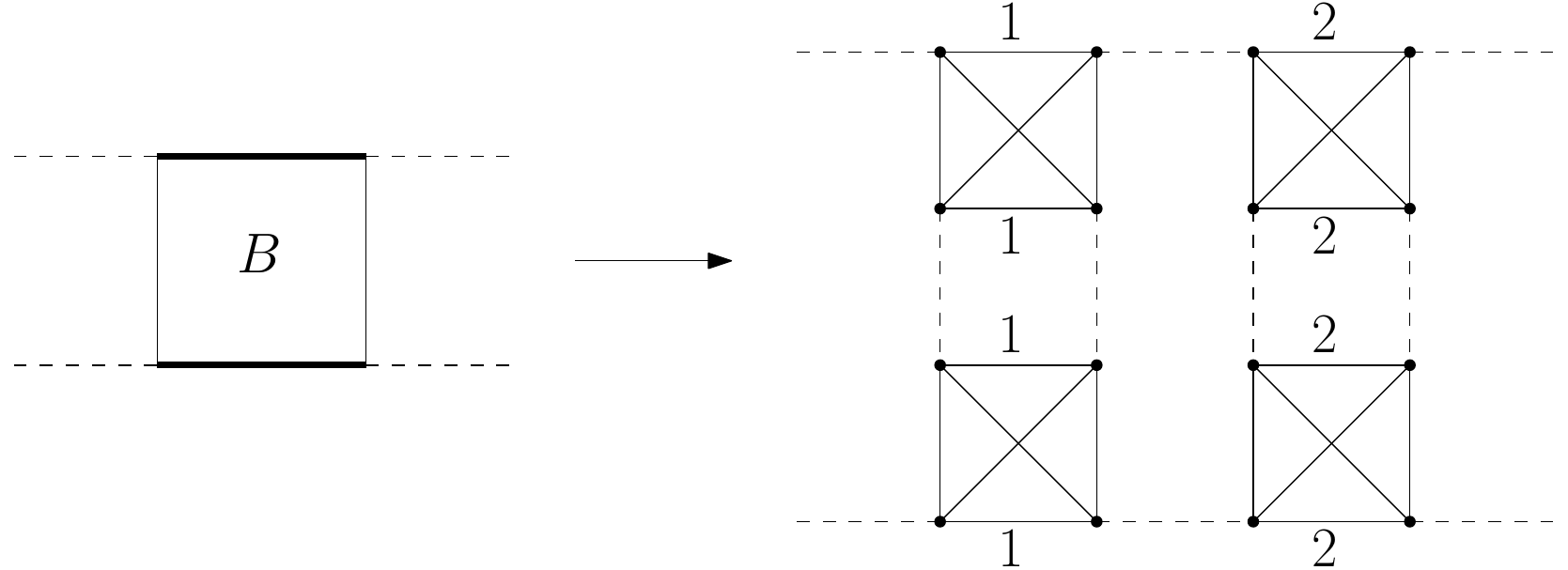}
\end{aligned}
\end{equation}

\begin{description}
\item[Removal of a non-separating dipole of color $i$]
There is either one or two faces of color $i$ incident to the dipole in $\cG$, whose structures can be as follows (we represent the paths of the faces with dotted edges)
\begin{equation}
\includegraphics[scale=.5, valign=c]{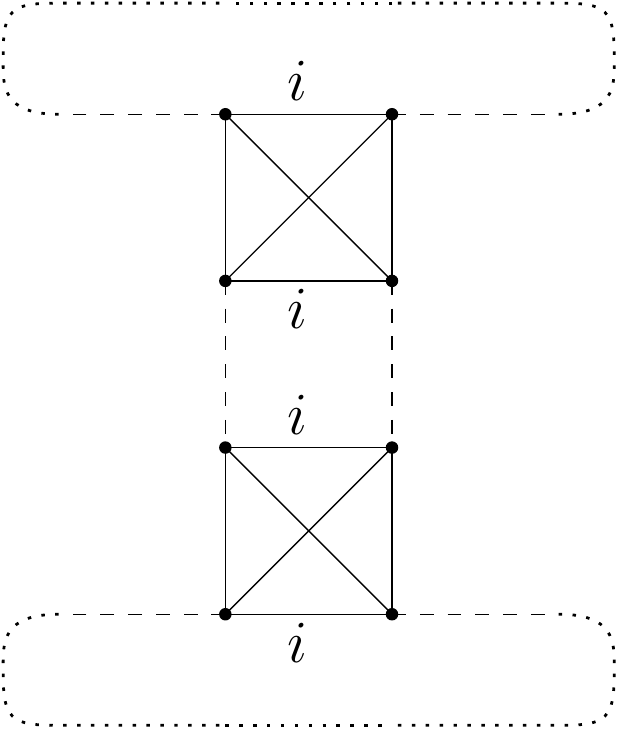} \qquad \includegraphics[scale=.5, valign=c]{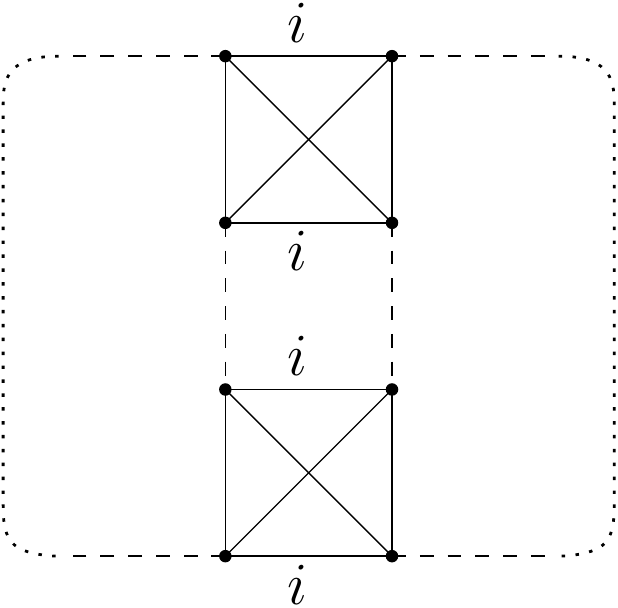}\qquad \includegraphics[scale=.5, valign=c]{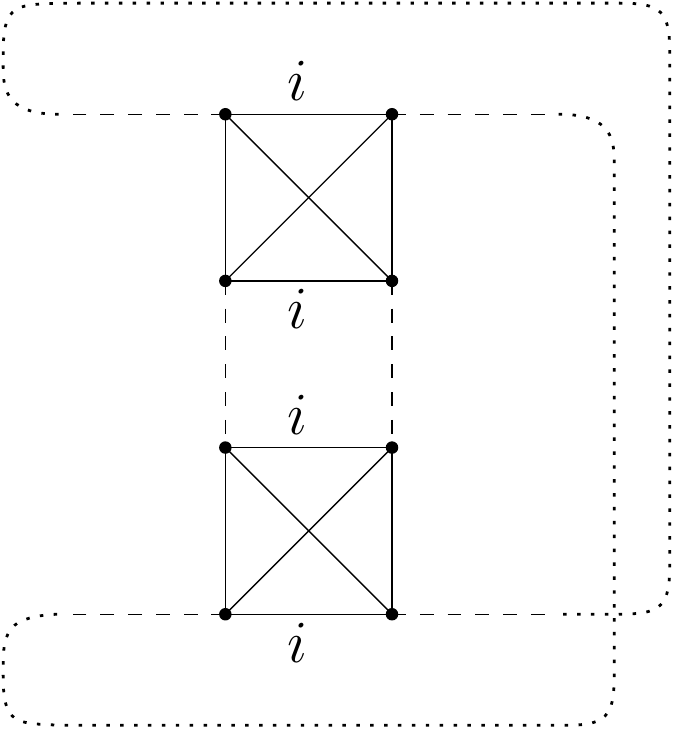}
\end{equation}
and the dipole has an internal face of color $i$. When removing the dipole, the internal face is deleted and one or two faces of color $i$ are formed. The faces of the other colors $j, k\neq i$ are unaffected by the move. The number of interactions decreases by $2$. Thus deleting a non-separating dipole gives $ -1 \geq \Delta \omega \geq -3$, where $\Delta \omega$ is the variation of the degree.

\item[Removal of a non-separating chain]
We have to distinguish two cases depending on whether the chain is broken or not.
\begin{itemize}
\item If the chain is broken then the structure of the faces incident to the chain is unchanged. Recall that since $\cG$ is a minimal realization of a scheme, a broken chains has exactly two dipoles of type I, in which case it is easy to check that the chain has exactly three internal faces. Therefore deleting the chain gives $\Delta \omega = -3$.
\item If the chain is not broken, the discussion is similar to case of the dipole above, and we have $ -1 \geq \Delta \omega \geq -3$.
\end{itemize}

\item[Removal of separating chains and separating dipoles]
In the following we will tackle the case of separating chains, however the discussion also holds for separating dipoles. If a chain is separating, then the edges of color 0 on either side form a 2-edge-cut,
\begin{equation}
\includegraphics[scale=.5, valign=c]{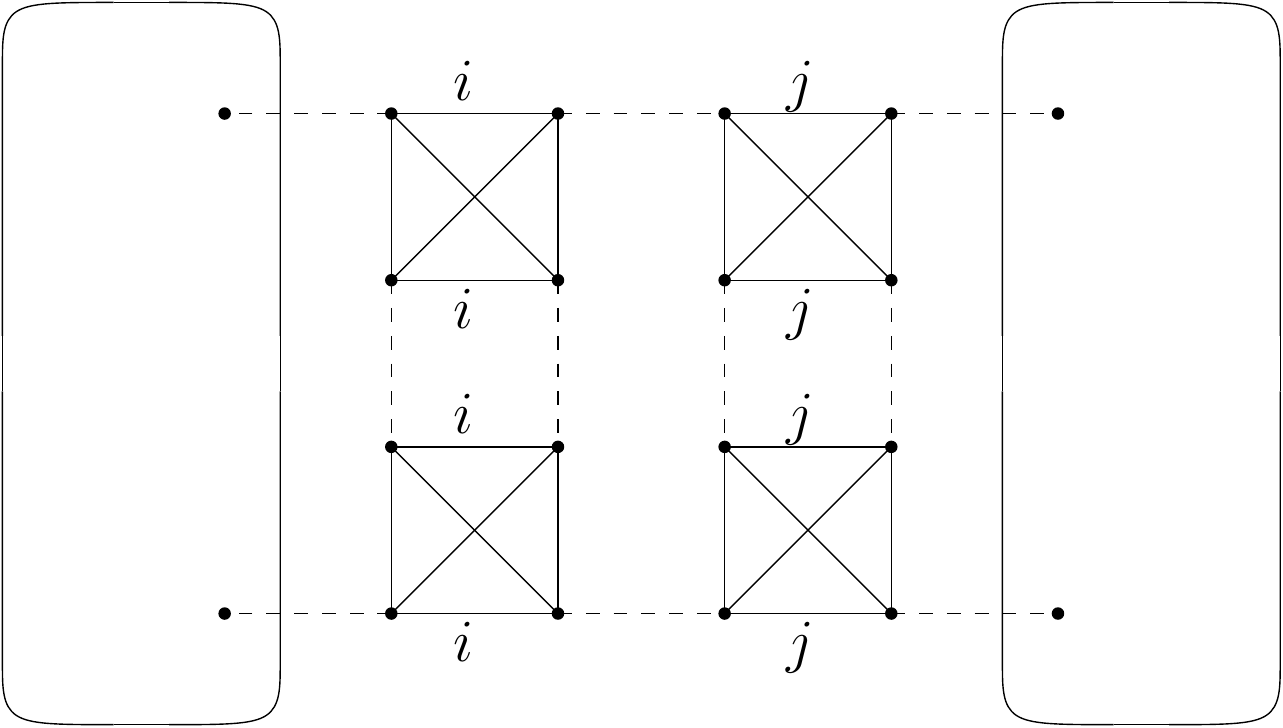}
\end{equation}
with $i,j\in\{1,2,3\}$. The chain removal produces two connected components $\cG_1$ and $\bar{\cG}_2$. It is straightforward to check that in either case $i=j$ and $i\neq j$, $\omega(\cG) = \omega(\cG_1) + \omega(\bar{\cG}_2)$.
\end{description}
\end{proof}

\subsection{Combinatorics of the skeleton graph $\cI(\cS)$}

To identify the vertex set of the skeleton graph $\cI(\cS)$, one first performs all dipole- and chain-vertex removals in $\cS$. This gives one 2-point graph $\cG^{(0)}$ and a collection of vacuum graphs $\bar{\cG}^{(1)}, \dotsc, \bar{\cG}^{(p)}$, which in turn gives a root vertex and $p$ additional vertices in $\cI(\cS)$.

\begin{lemma} \label{thm:SkeletonGraph}
For any scheme $\cS$, define by convention $\omega(\cI(\cS)) = \omega(\cS)$. The following properties hold.
\begin{enumerate}
\item\label{enum:Valency3} If $\bar{\cG}^{(r)}$ for $r\in\{1, \dotsc, p\}$ has vanishing degree, then the corresponding vertex in $\cI(\cS)$ has valency at least equal to 3.
\item\label{enum:RemoveNonSeparating} Let $\cT\subset \cI(\cS)$ be a spanning tree. Let $q$ be the number of edges of $\cI(\cS)$ which are not in $\cT$. Then, $\omega(\cT) \leq \omega(\cS)-q$.
\item\label{enum:SpanningTree} $\omega(\cT) = \omega(\cG^{(0)}) + \sum_{r=1}^p \omega(\bar{\cG}^{(r)})$, in other words, if $\cI(\cS)$ is a tree, then the degree of $\cS$ is the sum of the degrees of its components obtained by removing all chain-vertices.
\end{enumerate}
\end{lemma}

\begin{proof}
\ref{enum:Valency3}. If $\cG^{(r)}$ has degree zero and the corresponding vertex in $\cI(\cS)$ has valency 1, then $\cG^{(r)}$ is a melonic 2-point function, which cannot happen in schemes. If $\cG^{(r)}$ has degree zero and the corresponding vertex in $\cI(\cS)$ has valency 2, then $\cG^{(r)}$ is a chain, which is impossible in a scheme (it means that a chain-vertex was used in place of a non-maximal chain).

\ref{enum:RemoveNonSeparating}. The edges of $\cI(\cS)$ which are not in $\cT$ correspond to non-separating chain-vertices in $\cS$. One then concludes using Equation \eqref{NonSeparatingRemoval} from Lemma \ref{thm:ChainRemoval}.

\ref{enum:SpanningTree}. Denote $\cS_{\cT}$ the scheme obtained after removing the above $q$ non-separating chain-vertices. Its skeleton graph is $\cT$, i.e. $\cI(\cS_{\cT}) = \cT$, meaning that all its chain-vertices are separating. One concludes with Equation \eqref{SeparatingRemoval} from Lemma \ref{thm:ChainRemoval}.
\end{proof}

\subsection{Proof of Lemma~\ref{lemma:fst}}

Let $\cS$ be a scheme of degree $\omega>0$ and consider the notations of Lemma \ref{thm:SkeletonGraph}. Notice that a chain-vertex removal can decrease the number of chain-vertices by at most 3. Indeed, only two edges of color 0 are affected by the removal, which means that at most two pairs of chain-vertices can be joined after the removal to form new chain-vertices. This happens in the following situation,
\begin{equation}
\includegraphics[scale=.5, valign=c]{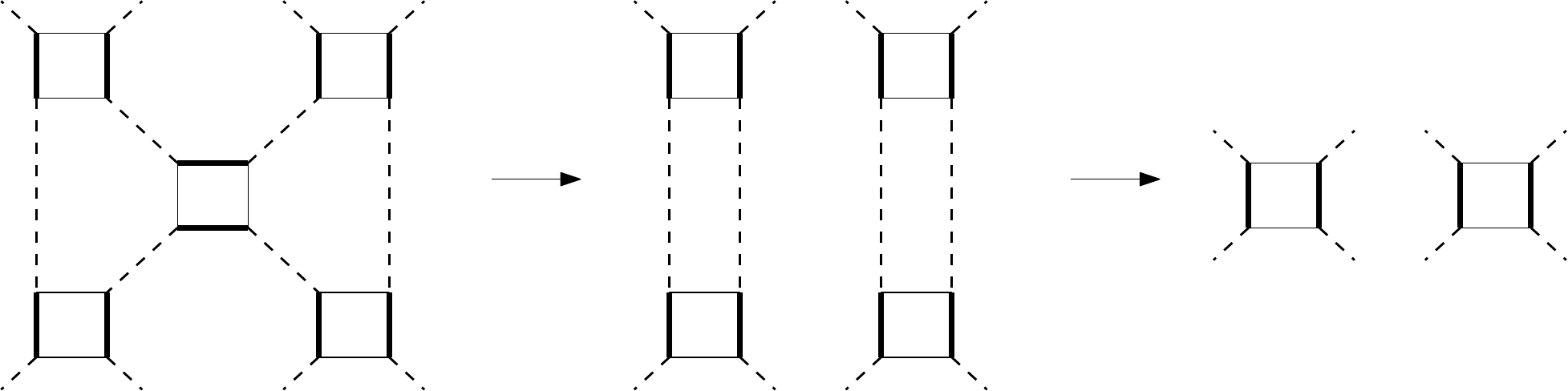}
\end{equation}
After the removal (first arrow), one obtains an object which is not a scheme, because two consecutive chain-vertices would account for two, consecutive, hence non-maximal, chains. The correct scheme is obtained by replacing the maximal chains with chain-vertices (second arrow).

Let $N(\cS)$ be the number of chain-vertices of $\cS$ and $N(\cT)$ the number of chain-vertices of the scheme $\cS_{\cT}$ obtained after removing the above $q$ non-separating chain-vertices from $\cS$. One finds $N(\cS) \leq N(\cT) + 3q$. A consequence of point \ref{enum:RemoveNonSeparating} in Lemma \ref{thm:SkeletonGraph} is that $q\leq \omega(\cS) - \omega(\cT) \leq \omega(\cS)$. Therefore 
\begin{equation} \label{BoundSchemeTree}
N(\cS) \leq N(\cT) + 3\omega(\cS).
\end{equation}

We say that $\bar{\cG}^{(r)}$, for $r\in\{1, \dotsc, p\}$, is \emph{tracked}, if it is incident to one of the non-separating chain-vertices (corresponding to the edges of $\cI(\cS)$ which are not in $\cT$). Denote $N_0^{t}$ the number of tracked components of degree 0, $N_0^{nt}$ the number of non-tracked components of degree 0, and $N_+$ the number of components of positive degree, all excluding $\cG^{(0)}$. By definition, there are fewer than $2q$ tracked components,
\begin{equation} \label{eq:bound_C0}
N_0^t\leq 2q.
\end{equation}
Since components of positive degree have degree at least $1/2$, we find from \ref{enum:SpanningTree} in Lemma \ref{thm:SkeletonGraph} that
\begin{equation} \label{eq:bound_pos}
\omega(\cT) \geq \frac{1}{2}N_+
\end{equation}
Together with \ref{enum:RemoveNonSeparating} from Lemma \ref{thm:SkeletonGraph}, and \eqref{eq:bound_C0}, this leads to 
\begin{equation} \label{PartialDegreeBound}
N^t_0 + N_+ \leq 2\omega(\cT)
\end{equation}
Notice that $N(\cT) = p+q$ is also the number of edges of $\cT$. Since it is a tree, the number of edges is the number of vertices minus one, hence the following relation
\begin{equation}
N(\cT) = N^{nt}_0 + N^t_0 + N_+ 
\label{eq:tree_euler}
\end{equation}
Finally, counting leaves and nodes of $\mathcal{T}$ weighted by their valency amounts to counting twice the number of edges of $\mathcal{T}$. Due to point \ref{enum:Valency3} in Lemma \ref{thm:SkeletonGraph}, we have
\begin{equation}
2N(\cT) \geq 3N^{nt}_0 + N^t_0 + N_+ + 1
\label{eq:tree_val}
\end{equation}
the additional one being due to the component $\cG^{(0)}$ (with valency at least 1). The previous two equations lead to $N(\cT) \leq 2(N_0^t+N_+)$. Together with \eqref{PartialDegreeBound}, $N(\cT) \leq 4\omega(\cT)$. From \eqref{BoundSchemeTree}, one gets
\begin{equation}
N(\cS) \leq 4\omega(\cT) + 3\omega(\cS) - 1 \leq 7\omega(\cS) -1,
\end{equation}
which proves Lemma \ref{lemma:fst}.
\subsection{Proof of Lemma~\ref{lemma:sec}}

Here we prove Lemma \ref{lemma:sec}, i.e. graphs of finite degree and with a finite number of isolated dipoles have a bounded number of bubbles. Our strategy is to show that there exist bounds on the number of faces of every degree, depending on the degree $\omega$ and the number of isolated dipoles $k$, i.e.
\begin{equation} \label{BoundFaceDegreeP}
F^{(p)}(\mathcal{G}) \leq \phi^{(p)}(\omega, k)
\end{equation}
where $F^{(p)}(\mathcal{G})$ is the number of faces of degree $p$.

There is exactly one face of color $i=1,2,3$ along each edge of color $0$. Therefore, denoting $F_i^{(p)}$ the number of faces with color $i$ and degree $p$, we have for each color $i$
\begin{equation}
\sum\limits_{p \geq 1} pF_i^{(p)}(\mathcal{G}) = E_0(\cG) = 2n(\cG)
\end{equation}
Using the degree formula~\eqref{eq:deg}, $n(\cG)$ can be eliminated to give
\begin{equation} 
\sum_{p \geq 5} (p-4)F^{(p)}(\mathcal{G}) = 4(\omega-3) + 3 F^{(1)}(\mathcal{G}) + 2 F^{(2)}(\mathcal{G}) + F^{(3)}(\mathcal{G})
\label{eq:face_bound}
\end{equation}
We have written this equation so that both sides come with positive coefficients (except the irrelevant $-12$).
\begin{itemize}
\item Therefore if one can bound $F^{(1)}(\mathcal{G})$, $F^{(2)}(\mathcal{G})$, $F^{(3)}(\mathcal{G})$, i.e. if one can prove \eqref{BoundFaceDegreeP} for $p=1,2,3$, then it automatically gives the bound \eqref{BoundFaceDegreeP} for $p\geq 5$.
\item It remains to prove the bound independently for $p=4$, since $F^{(4)}$ does not appear in \eqref{eq:face_bound}. There is indeed a risk that graphs with an arbitrarily large number of faces of degree 4 exist, without having these faces affecting the degree. However this is not the case, as we will prove.
\end{itemize}
In our analysis, we will distinguish self-intersecting and non-self-intersecting faces. A face is said to be \emph{self-intersecting} if it visits the same bubble more than once. Otherwise it is called \emph{non-self-intersecting} (n.s.i.).

\paragraph{Faces of degree $1$\\}
If a graph has a face of degree $1$, then it has a tadpole,
\begin{equation}
\begin{array}{c} \includegraphics[scale=.3]{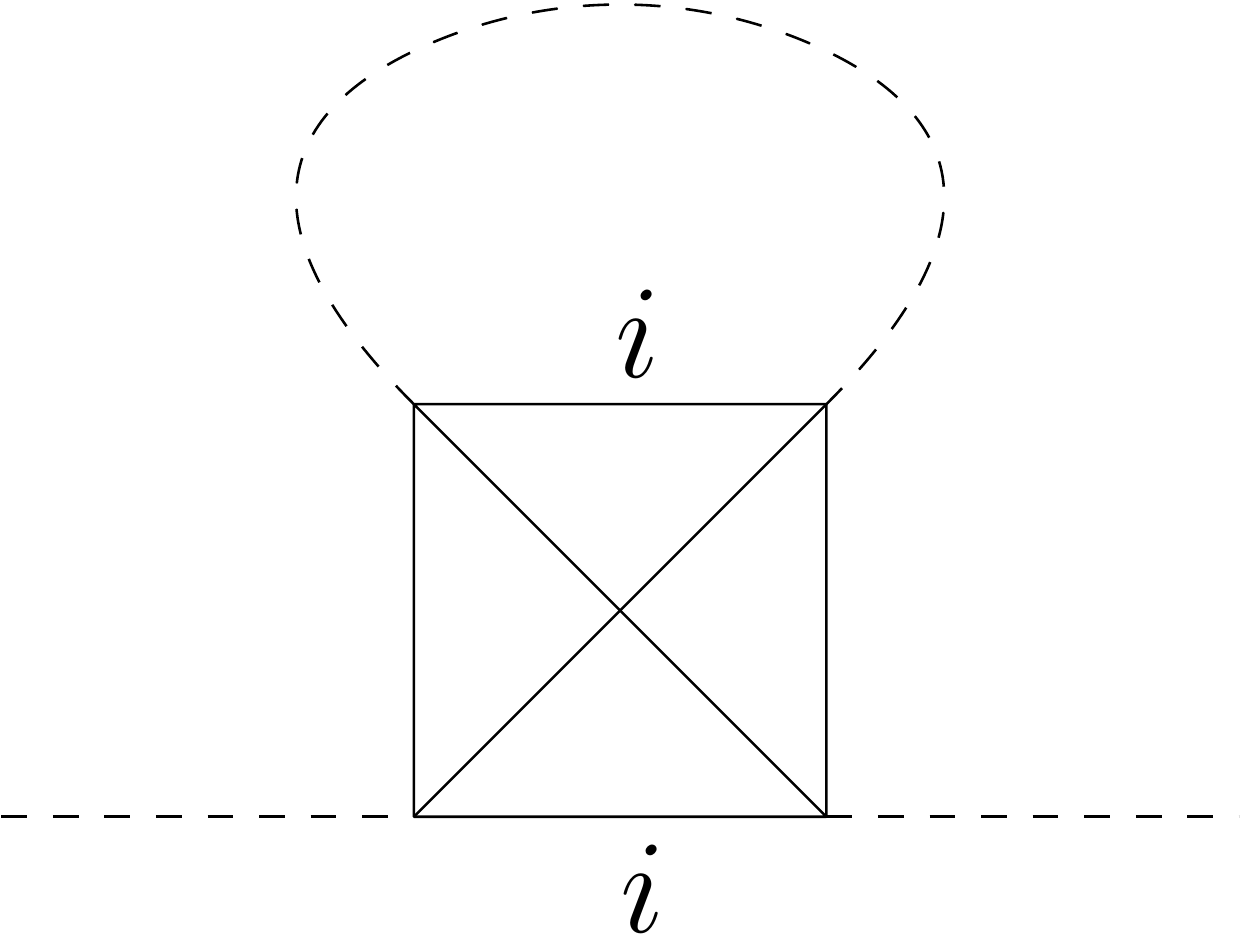} \end{array}
\end{equation}
When that bubble and the tadpole are removed and replaced with an edge of color 0, the degree decreases by $\frac{1}{2}$. Therefore, there are at most $2\omega$ tadpoles in a graph $\cC$ of degree $\omega$ i.e. $F^{(1)}(\cG)\leq 2\omega$.

\paragraph{Faces of degree $2$\\}
First notice that there is a single graph which has a face of degree 2 which is self-intersecting. It consists in a bubble with two tadpoles.

If $\cG$ has a face of degree $2$ which is n.s.i., then it is a dipole. Either it is a dipole which is isolated and there are $k$ of them, or it is a non-isolated dipole. In the latter case, we perform the dipole removal,
\begin{equation} \label{NonIsolatedDipoleRemoval}
    \begin{array}{c} \includegraphics[scale=.45]{NonIsolatedDipoles.pdf}\end{array} \quad \to \quad \begin{array}{c} \includegraphics[scale=.45]{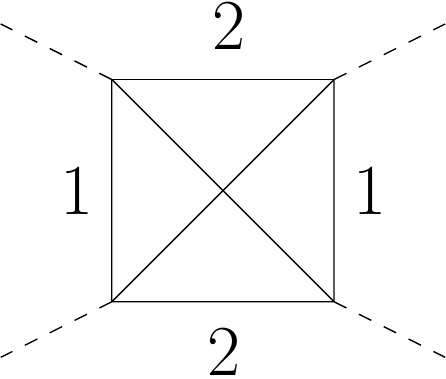}\end{array}
\end{equation}
A non-isolated dipole is always non-separating, therefore one obtains a connected graph $\cG'$ whose degree satisfies $\omega(\cG')<\omega(\cG)$. We thus proceed by induction on the degree.

Denote $F^{(2)}_{\text{NI}}(\cG)$ the number of non-isolated dipoles of $\cG$. If $\cG$ is melonic, then $F^{(2)}_{\text{NI}}(\cG)=0$. Let $\omega>0$ and assume that for all $\cG'$ of degree $\omega'<\omega$ with $k'$ isolated dipoles we have the bound $F^{(2)}_{\text{NI}}(\cG) \leq \phi^{(2)}(\omega',k')$. Let $\cG$ of degree $\omega$ with $k$ isolated dipoles. We perform \eqref{NonIsolatedDipoleRemoval} and obtain $\cG'$ of degree $\omega'<\omega$. We track the changes in the number of dipoles:
\begin{itemize}
    \item The number of non-isolated dipoles cannot increase (note that it may remain unchanged).
    \item Moreover the bubble on the RHS of \eqref{NonIsolatedDipoleRemoval} can belong to at most one isolated dipole, therefore the number of isolated dipoles of $\cG$ satisfies $k'\leq k+1$. 
\end{itemize}
From the induction hypothesis we thus find
\begin{equation}
    F^{(2)}_{\text{NI}}(\cG) \leq F^{(2)}_{\text{NI}}(\cG') \leq \max_{\substack{\omega'<\omega\\ k'\leq k+1}} \phi^{(2)}(\omega',k').
\end{equation}
The RHS defines $\phi^{(2)}(\omega,k)$.


\paragraph{Faces of degree $3$.} We distinguish the cases where $\cG$ has a face of degree which is self-intersecting or not.

\subparagraph{Self-intersecting faces of degree 3\\}
If a face of degree $3$ is self-intersecting, then it necessarily has the following structure,
\begin{equation}
\begin{array}{c} \includegraphics[scale=0.45]{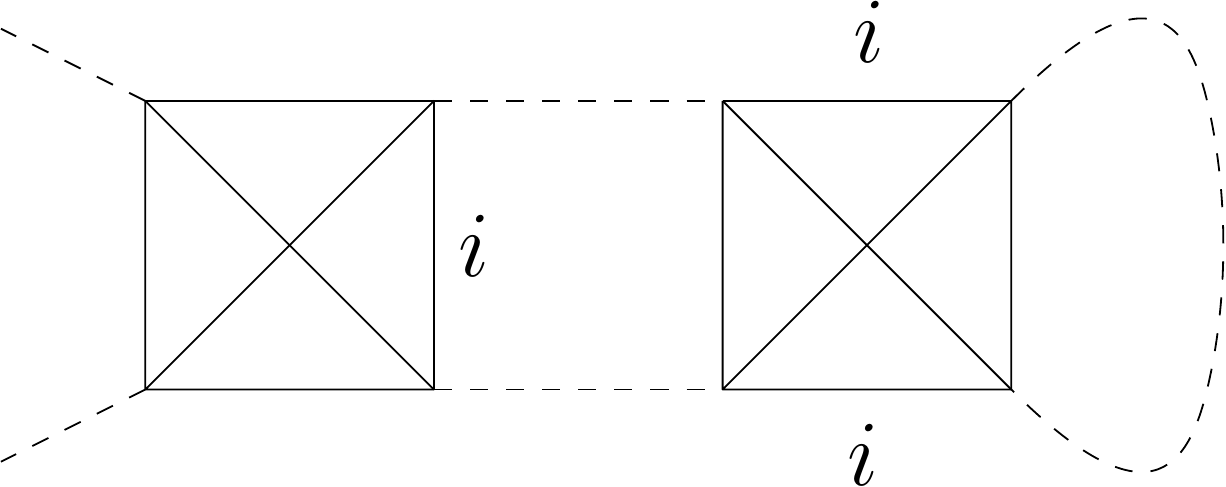} \end{array}
\end{equation}
It is a 2-point function whose removal decreases the degree by 1. Thus there are at most $\omega$ such faces in the graph, $F^{(3)}_{\text{s.i.}}(\mathcal{G})\leq \omega$.

\subparagraph{Non-self-intersecting faces of length $3$\\}
An n.s.i. face of degree 3 and color $1$ necessarily has the following structure,
\begin{equation} \label{F3nsi}
\begin{array}{c} \includegraphics[scale=0.4]{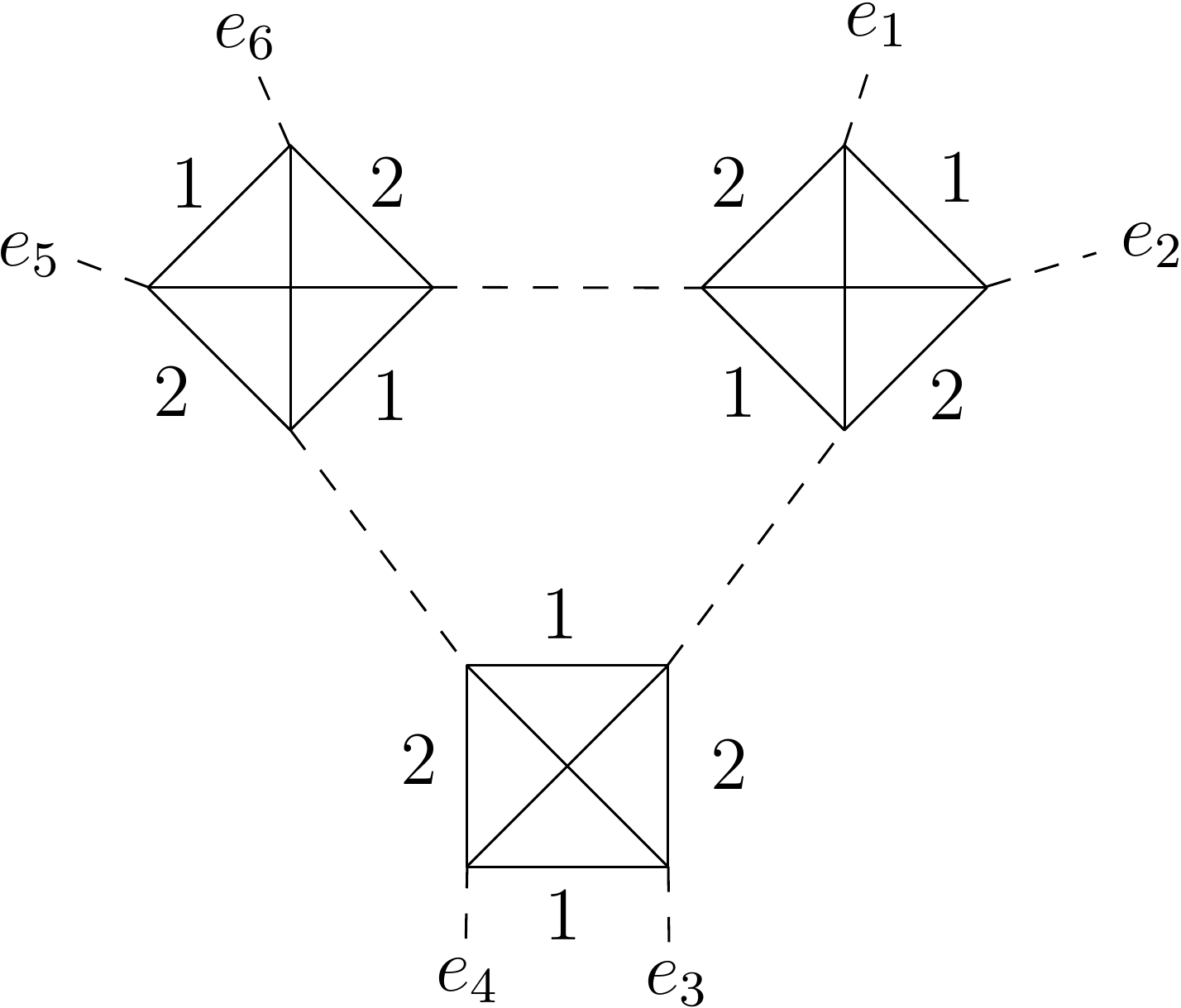} \end{array}
\end{equation}
There are at most $3$ different faces of each color passing along the edges $e_1, \dotsc, e_6$. We denote $f_1, f_2, f_3 \leq 3$ their numbers (note that $2$ and $3$ play symmetric roles). We want to show that it is always possible to remove this subgraph from $\cG$ and reglue the $6$ half-edges while decreasing the degree of the graph. We will use the following lemmas.

\begin{lemma} \label{thm:Variations}
\begin{enumerate}
\item Consider a n.s.i. face of degree 3 in $\cG\in\mathbb{G}_{O(N)^3}$, as in \eqref{F3nsi}, and a move on $\cG$ which removes the 3 bubbles of that face and reconnects the half-edges $e_1, \dotsc, e_6$ pairwise (the resulting graph may not be connected). Denote $\Delta F^{(d)}$ the variation of the number of faces of degree $d$ and $\Delta k$ the variation of the number of dipoles. Then
\begin{equation}
\Delta F^{(d)} \geq -9 \qquad \text{and} \qquad \Delta k \leq 9.
\end{equation}
\item Let $\cG\in\mathbb{G}_{O(N)^3}$ with $k$ dipoles. Let $e, e'$ be two edges of color 0 in $\cG$ forming a 2-edge-cut, and perform a flip
\begin{equation}
\begin{array}{c} \includegraphics[scale=.3]{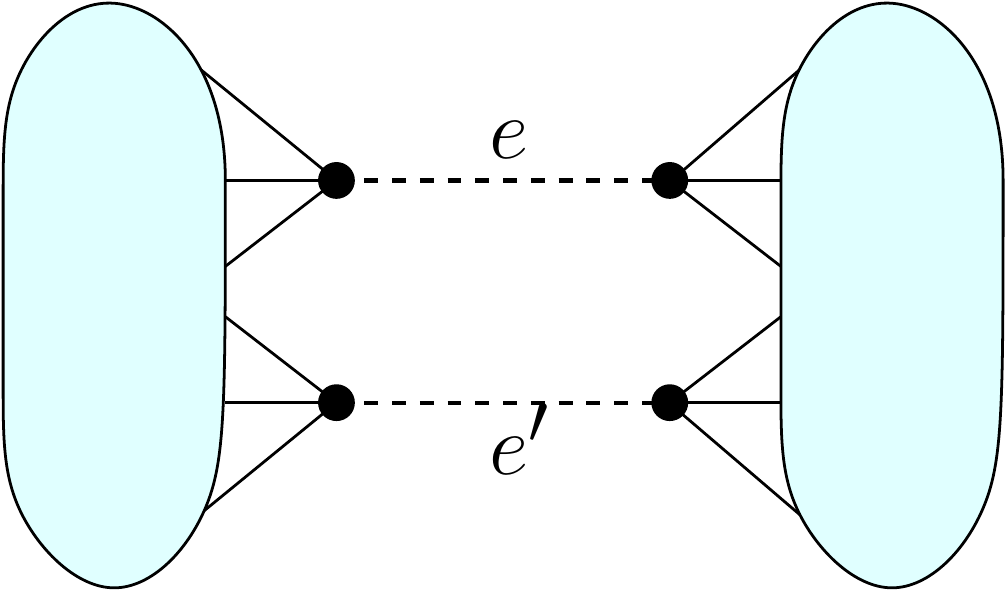} \end{array} \qquad \to \qquad \mathcal{G}_L = \begin{array}{c} \includegraphics[scale=.3]{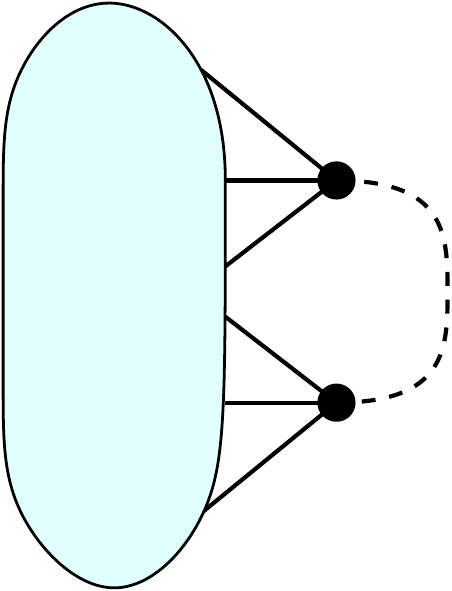} \end{array} \quad \cup \quad \mathcal{G}_R = \begin{array}{c} \includegraphics[scale=.3]{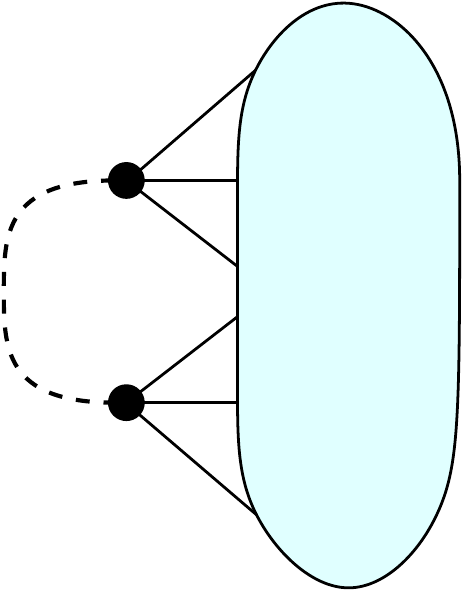} \end{array}
\end{equation}
which gives rise to two connected components $\cG_L$, $\cG_R$. Then the number of dipoles of each is bounded as $k(\cG_\alpha) \leq k+3$, $\alpha=L, R$, and the number of faces of degree $d$ as $F^{(d)}(\cG_L) + F^{(d)}(\cG_R) \geq F^{(d)}(\cG) - 3$.
\end{enumerate} 
\end{lemma}

This lemma contains only particular cases of a more general principle. If $\mathcal{H}$ is a subgraph, then the number of incident faces of each degree can be bounded as a function of $\mathcal{H}$ and not the graphs it is contained in. Then when one removes $\mathcal{H}\subset \cG$ in some way, the variations of the number of faces of each degree is bounded independently of $\cG$.

\begin{proof}
Along each edge there are exactly 3 faces, one of each color.
\begin{enumerate}
\item Before the move, there is at most 9 different faces going along the edges $e_1, \dotsc, e_6$. If they all have degree $d$ and none of this degree are created by the move we have $\Delta F^{(d)} =-9$, which is the extremal case.

After the move, the pairing of the half-edges gives rise to 3 edges of color 0. If they have no faces in common and all of them are dipoles, this gives 9 dipoles. If no dipoles are destroyed in the move, this gives $\Delta k = 9$.

\item For each color 1, 2, 3, it is the same face going along $e$ and $e'$ and the move splits each of them into two. If those three faces in $\cG$ were of degree $d$, and they are split in $\cG_L, \cG_R$ into faces of different degrees, then $F^{(d)}(\cG_L) + F^{(d)}(\cG_R) = F^{(d)}(\cG) - 3$, which is the extremal case.

\end{enumerate}
\end{proof}

\begin{lemma} \label{thm:TypeIMove}
Assume $\cG$ has a n.s.i. of degree 3 such that $f_2\leq 3$ and $f_3\leq 2$. Further assume that the following move gives a graph $\cG'$ connected
\begin{equation}
\begin{array}{c} \includegraphics[scale=.65]{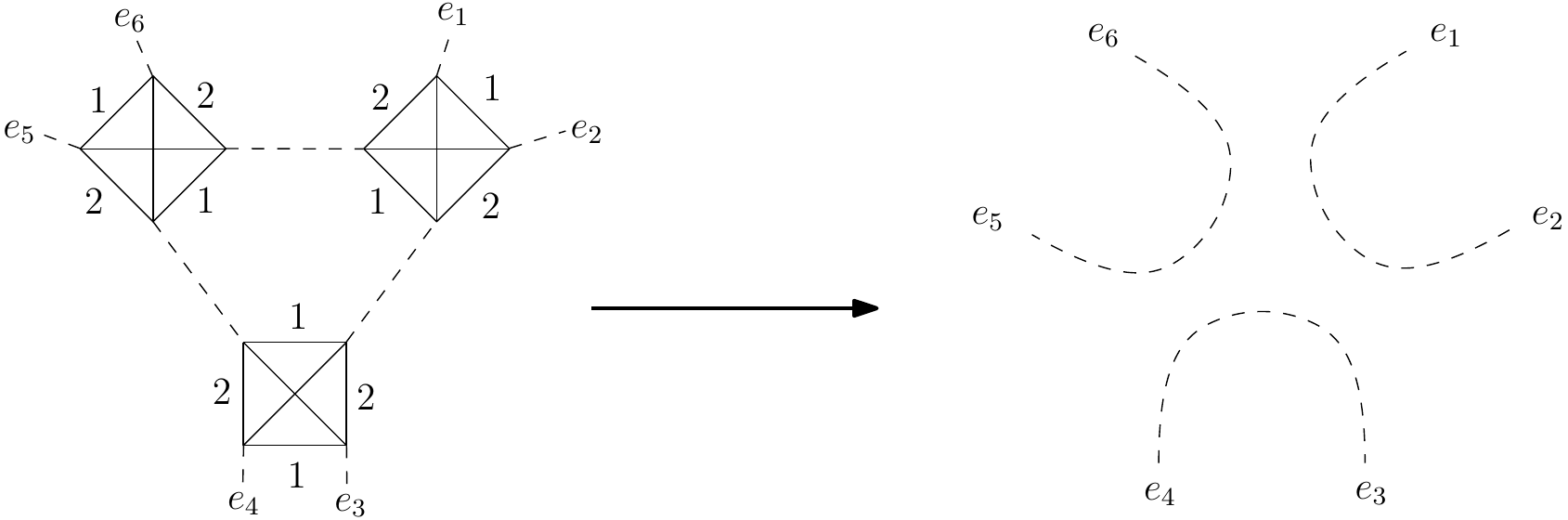} \end{array}
\end{equation}
then the degree changes as $\Delta \omega = \omega(\cG') - \omega(\cG) \leq -1/2$.
\end{lemma}

\begin{proof} 
The variation in the degree after the move is given by 
\begin{equation}
\Delta \omega = -\frac{7}{2} - \Delta f_1 - \Delta f_2 - \Delta f_3
\end{equation}
where $\Delta f_{i}$ is the variation of the number of faces of color $i$ which go along $e_1, \dotsc, e_6$ in $\cG$. In the lemma, the move does not change $f_1$. Moreover, $\Delta f_2\geq -2$ and $\Delta f_3 \geq -1$, since there is at least one face of color $2$ and one of color $3$ going along the edges on the configuration on the right. Thus $\Delta \omega \leq -\frac{1}{2}$.
\end{proof}

%
%
%


We can now prove \eqref{BoundFaceDegreeP} for n.s.i. faces of degree 3. We proceed by induction on the degree. At degree 0, all faces have even degree so $F^{(3)}(\cG)=0$.

Let $\omega>0$ and assume that there exists a bound $F^{(3)}_{\text{n.s.i.}}(\cG')\leq \phi^{(3)}_{\text{n.s.i.}}(\omega', k')$ for all graphs $\cG'$ of degree $\omega'<\omega$. Let $\cG\in\mathbb{G}_{O(N)^3}$ have degree $\omega$.

We first consider the cases where a pair of edges $e_i, e_j$ forms a 2-edge-cut. Without loss of generality, we will consider the pairs to be $\{e_1, e_j\}$ for $j=2, \dotsc, 6$. Clearly, the case of $\{e_1, e_5\}$ is a 2-cut is equivalent to $\{e_1, e_3\}$ by exchanging the colors 2 and 3. Same for $\{e_1, e_6\}$ which is equivalent to $\{e_1, e_4\}$. It is therefore enough to consider the cases where $\{e_1, e_j\}$ is a 2-cut for $j=2,3,4$.

The first step is to perform the cut and obtain two connected components $\cG_1, \cG_2$ and $\omega(\cG) = \omega(\cG_1) + \omega(\cG_2)$. We denote $\cG_1$ the component which inherits the n.s.i. face of degree 3. From Lemma \ref{thm:Variations}
\begin{equation}
F^{(3)}_{\text{n.s.i.}}(\cG) \leq  F^{(3)}_{\text{n.s.i.}}(\cG_1) + F^{(3)}_{\text{n.s.i.}}(\cG_2) + 3
\end{equation}
$\cG_1$ is not melonic (it has a face of degree 3), so $\omega(\cG_1)>0$. This gives $\omega(\cG_2)<\omega(\cG)$ and thus by the induction hypothesis 
\begin{equation}
F^{(3)}_{\text{n.s.i.}}(\cG_2) \leq \phi^{(3)}_{\text{n.s.i.}}(\omega(\cG_2), k(\cG_2)) \leq \max_{k_2\leq k+3} \phi^{(3)}_{\text{n.s.i.}}(\omega(\cG_2), k_2)
\end{equation}
To bound $F^{(3)}_{\text{n.s.i.}}(\cG_1)$, we will notice from Lemma \ref{thm:Variations} that $\cG_1$ has at most $k+3$ dipoles (hence isolated dipoles) and distinguish the following cases.

\begin{description}
\item[$\{e_1, e_2\}$ is a 2-cut (or form a single edge)] 
$\cG_1$ has a tadpole, 
\begin{equation}
\begin{array}{c} \includegraphics[scale=.6]{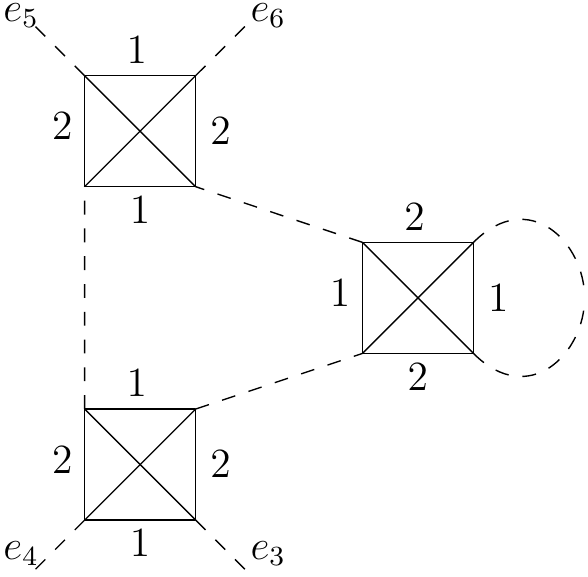} \end{array}
\end{equation}
which we remove to get $\cG_1'$ whose degree is $\omega(\cG_1') = \omega(\cG_1)-1/2$. From Lemma \ref{thm:Variations}, $\cG_1'$ has at most $k+6$ isolated dipoles ($k+3$ in fact holds because in this particular case, $\cG_1$ has at most all the dipoles of $\cG$). Thus,
\begin{equation}
F^{(3)}_{\text{n.s.i.}}(\cG_1) \leq F^{(3)}_{\text{n.s.i.}}(\cG_1') + 3 \leq \max_{k_1\leq k+6} \phi^{(3)}_{n.s.i.}(\omega-\frac{1}{2}, k_1) + 3
\end{equation}

\item[$\{e_1, e_3\}$ is a 2-cut (or form a single edge)] 
$\cG_1$ is as follows
\begin{equation}
\begin{array}{c} \includegraphics[scale=.6]{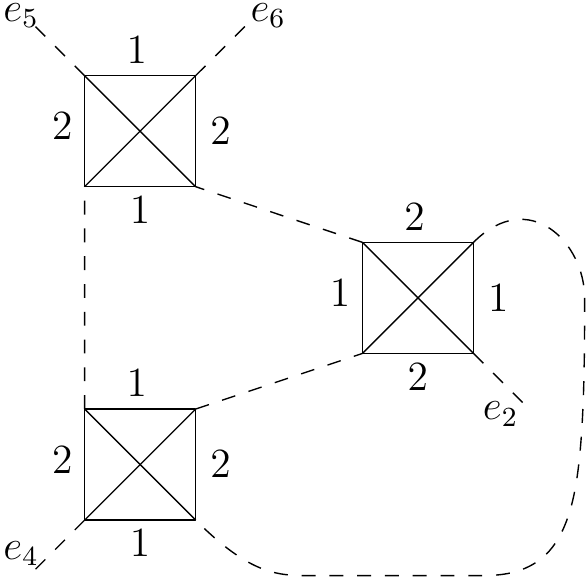} \end{array}
\end{equation}
so that $f_2, f_3\leq 2$. We can assume that $\{e_5, e_6\}$ is not a 2-edge-cut, since it would be equivalent to the case where $\{e_1, e_2\}$ is a 2-cut if it were. Therefore the move from Lemma \ref{thm:TypeIMove} does not disconnect $\cG_1$ and one obtains $\cG_1'$ of degree $\omega(\cG_1')\leq \omega(\cG_1)-1/2$. From Lemma \ref{thm:Variations}, it has at most 9 more isolated dipoles than $\cG_1$ and 9 n.s.i. faces of degree 3 less than $\cG_1$. This gives
\begin{equation}
F^{(3)}_{\text{n.s.i.}}(\cG_1) \leq F^{(3)}_{\text{n.s.i.}}(\cG_1') + 9 \leq \max_{\substack{\omega_1\leq \omega-1/2\\ k_1\leq k+12}} \phi^{(3)}_{\text{n.s.i.}}(\omega_1, k_1) + 9
\end{equation}

\item[$\{e_1, e_4\}$ is a 2-cut (or form a single edge)] 
$\cG_1$ is as follows
\begin{equation}
\begin{array}{c} \includegraphics[scale=.6]{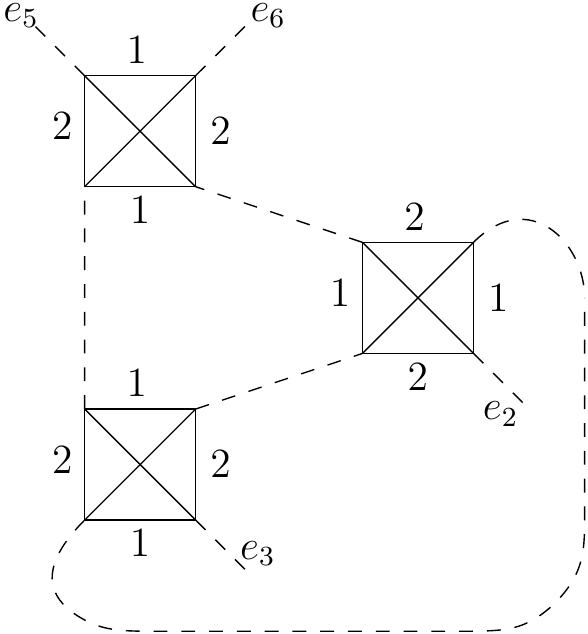} \end{array}
\end{equation}
so that $f_2\leq 2, f_3\leq 3$. The rest is therefore identical to the case where $\{e_1, e_3\}$ is a 2-cut above.
\end{description}

We now assume that no pair $\{e_i, e_j\}$ forms a 2-cut and distinguish two cases.
\begin{description}
\item[$f_2 \leq 3$ and $f_3\leq 2$ (or the other way around)] 
We can directly perform the move from Lemma \ref{thm:TypeIMove} and obtain $\cG'$ which is connected and of degree $\omega(\cG')\leq \omega-1/2$. From Lemma \ref{thm:Variations} we further get
\begin{equation}
F^{(3)}_{\text{n.s.i.}}(\cG) \leq F^{(3)}_{\text{n.s.i.}}(\cG') + 9 \leq \max_{\substack{\omega'\leq \omega-1/2\\ k'\leq k+9}} \phi^{(3)}_{\text{n.s.i.}}(\omega', k') + 9
\end{equation}

\item[$f_2 = f_3 = 3$]
In this case, the move from Lemma \ref{thm:TypeIMove} would only guarantee $\Delta \omega \leq 1/2$ which is not enough for our purposes. Instead, we exchange some of the half-edges as follows
\begin{equation}
\begin{array}{c} \includegraphics[scale=0.65]{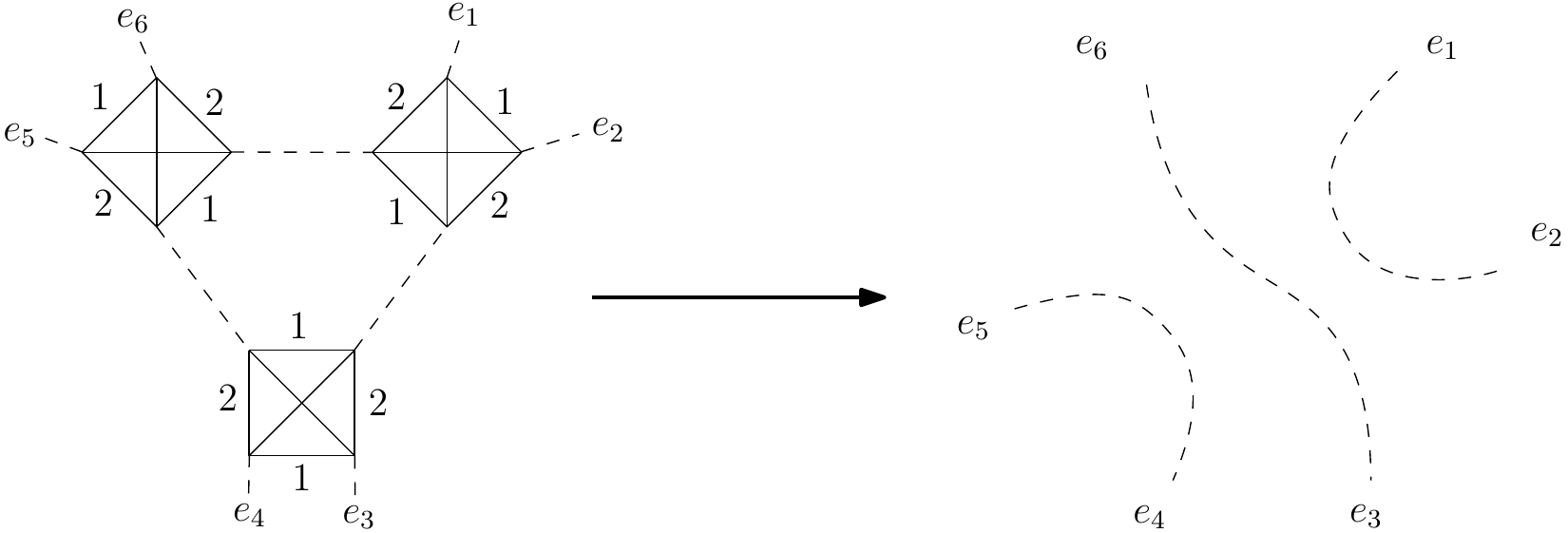} \end{array}
\end{equation}
which gives exactly $\Delta f_2 = \Delta f_3 = -1$. Moreover, if $f_1=3$, then $\Delta f_1=-1$ too, and if $f_1\geq 2$ then $\Delta f_1\geq -1$ because there is still at least one face on the configuration of the right hand side. Overall, $\Delta f_1 \geq -1$ and thus $\Delta \omega \leq -1/2$. We then conclude as previously.
\end{description}

%
%
%
\paragraph{Faces of degree $4$\\}

We have so far proved that at fixed degree $\omega$ and fixed number of isolated dipoles $k$, there is a bound on the number of faces of every degree $d$, except for $d=4$ which does not enter \eqref{eq:face_bound}. So {\it a priori}, there could be an unbounded number of faces of degree 4. However, for this to be possible they would have to be at arbitrarily large distance of any faces of degree $d\neq 4$ (since a bubble has a finite number of other bubbles at finite distance). Here the distance between two bubbles is the minimal number of edges of color 0 to go from any vertex of the first bubble to any vertex of the second bubble.

Similarly as for faces of degree $2$ and $3$, we distinguish two cases according to whether the face is self-intersecting or not.

\subparagraph{Self-intersecting faces of degree $4$\\}

Since a tetrahedral bubble has two edges of each color, a face can pass through a bubble at most twice. Therefore a self-intersecting face of degree $4$  has either 2 or 3 bubbles. If it is 2, then $\cG\in\mathbb{G}_{O(N)^3}$ has in fact only those two bubbles and there is a finite number of such graphs.

A self-intersecting face of degree 4 which goes along three bubbles can have a tadpole as follows (up to color permutations),
\begin{equation}
\begin{array}{c} \includegraphics[scale=.6]{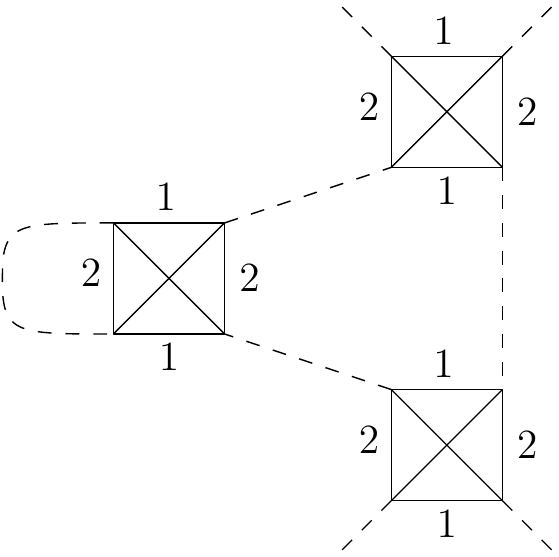} \end{array}
\end{equation}
the number of which we know to be bounded at fixed degree, or be as follows (up to color permutations)
\begin{equation}
\begin{array}{c} \includegraphics[scale=.45]{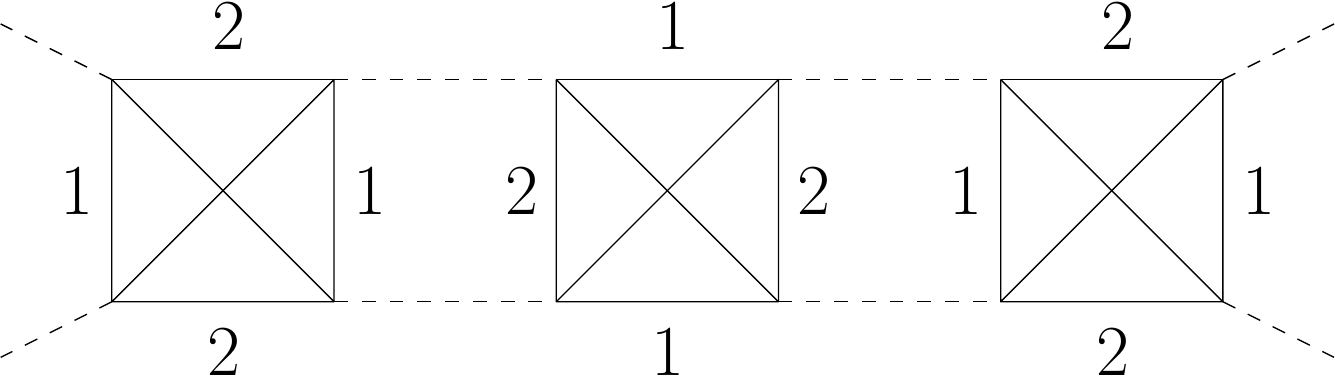} \end{array}
\end{equation}

\begin{itemize}
\item If the graph has a 2-cut as follows
\begin{equation}
\begin{array}{c} \includegraphics[scale=.45]{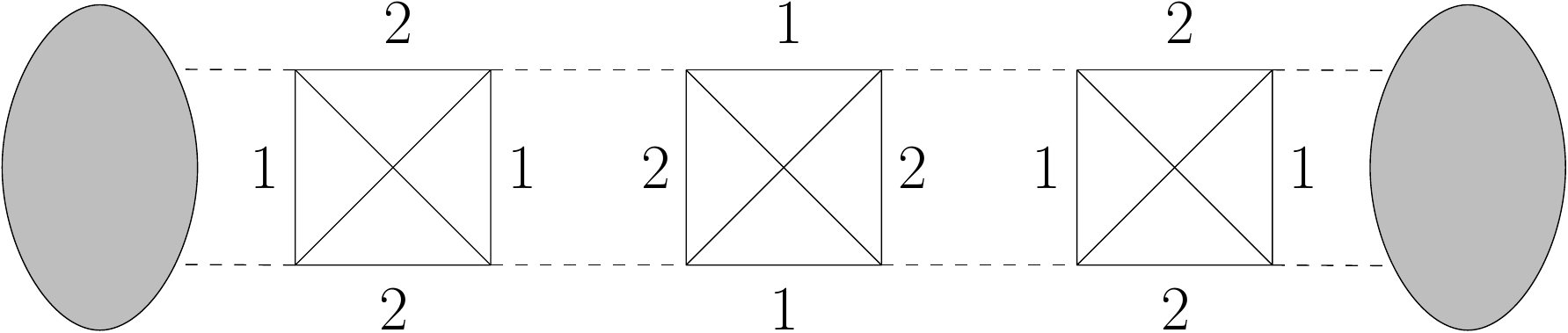} \end{array}
\end{equation}
we transform it to $\begin{array}{c} \includegraphics[scale=.25]{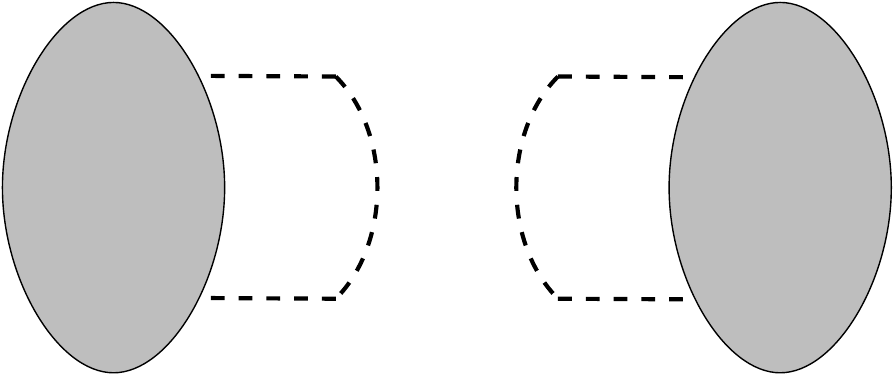} \end{array}$ and denote $\cG_L, \cG_R$ the two connected components. We find that
\begin{equation}
\omega(\cG) = \omega(\cG_L) + \omega(\cG_R) + \frac{3}{2}
\end{equation}
which means that $\omega(\cG_{L,R}) <\omega(\cG)$ so that the induction hypothesis can be applied to $\cG_{L, R}$. By applying Lemma \ref{thm:Variations} on both sides of the subgraph we find
\begin{equation}
F^{(4)}(\cG) \leq F^{(4)}(\cG_L) + F^{(4)}(\cG_R) + 7 \leq 2 \max_{\substack{\omega'\leq \omega-3/2\\ k'\leq k+3}} \phi^{(4)}(\omega',k') + 7
\end{equation}

\item Else, we perform the same move but the new graph $\cG'$ is connected and $\omega(\cG') \leq \omega(\cG)-5/2$. By adapting the arguments of Lemma \ref{thm:Variations}, we also find $F^{(4)}(\cG') \geq F^{(4)}(\cG) - 10$ and $k(\cG') \leq k(\cG)+10$. This leads to
\begin{equation}
F^{(4)}(\cG) \leq \max_{\substack{\omega'\leq \omega-5/2\\k'\leq k+10}} \phi^{(4)}(\omega', k') + 10
\end{equation}
\end{itemize} 

\subparagraph{Non-self-intersecting faces of degree $4$\\}
\label{par:nsi_f4}


Let $\cG\in\mathbb{G}_{O(N)^3}$ have a n.s.i. face of degree 4 and $B_{\text{exc}}\subset \cG$ the set of bubbles incident to at least one face of degree $d\neq 4$ or one self-intersecting face of degree 4. We have shown that
\begin{equation}
|B_{\text{exc}}| \leq \beta(\omega, k)
\end{equation}
We now want to show that $\cG$ cannot have bubbles at arbitrarily large distance from $B_{\text{exc}}$. Since the number of bubbles at finite distance of $B_{\text{exc}}$ is itself finite (because bubbles have degree 4), this will prove Lemma \ref{lemma:sec}. To do so, we will show that there are only a finite number of graphs where a bubble is at distance 3 or more from $B_{\text{exc}}$.

We denote $A, B, C, D$ the four bubbles of a n.s.i. face of degree 4
\begin{equation}
\begin{array}{c} \includegraphics[scale=.7]{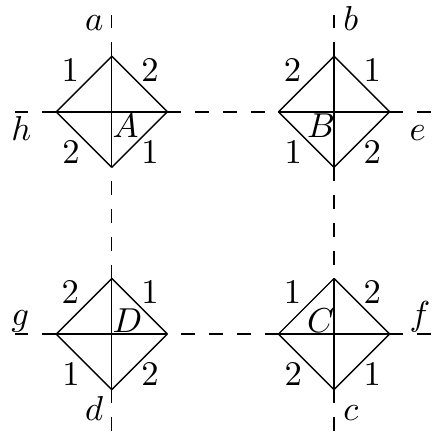} \end{array}
\end{equation}
and consider that \emph{A is at distance 3 of $B_{\text{exc}}$}. In particular, it is incident to only n.s.i. faces of degree 4. The half-edge $a$ cannot be connected to $b$ as it would form a dipole, nor to $d$ for the same reason. If $a$ is connected to $c$, then, in order for the faces of color 2 and 3 to have degree 4, it is necessary to connect $b$ to $d$. This leaves a subgraph which we replace with a single bubble as follows
\begin{equation}
\begin{array}{c} \includegraphics[scale=.7]{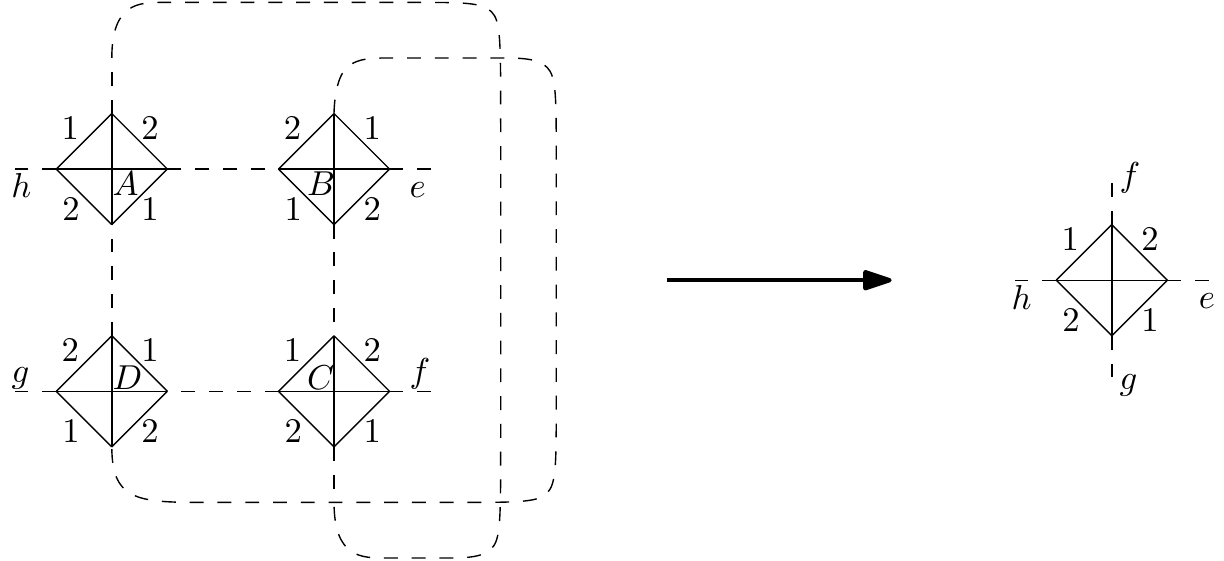} \end{array}
\end{equation}
This does not change any of the faces which go along $e, f, g, h$. Therefore the new graph $\cG'$ (which is connected) has only three faces less than $\cG$ and has degree $\omega(\cG') = \omega(\cG) -3/2$. Again adapting the arguments of Lemma \ref{thm:Variations}, we see that the number of faces of degree 4 cannot decrease by more than 6 (this is the largest number of faces which can go through $e,f,g,h$), the number of dipoles cannot increase by more than 6 (for the same reason). This gives
\begin{equation}
F^{(4)}(\cG) \leq \max_{k'\leq k+6} \phi^{(4)}(\omega-\frac{3}{2}, k') + 6
\end{equation}

We can now consider the case where $a$ is connected to another bubble,
\begin{equation}
\begin{array}{c} \includegraphics[scale=.7]{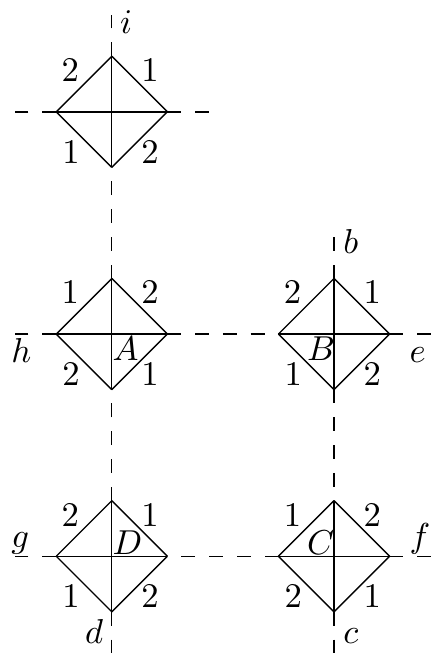} \end{array}
\end{equation}
The face of color 3 which goes along $d$ and $i$ must be of degree 4. The half-edges $d, i$ must therefore be connected to two vertices which are themselves connected by an edge of color 3. There is no such edge available in the subgraph, so a new bubble must be added,
\begin{equation}
\begin{array}{c} \includegraphics[scale=.7]{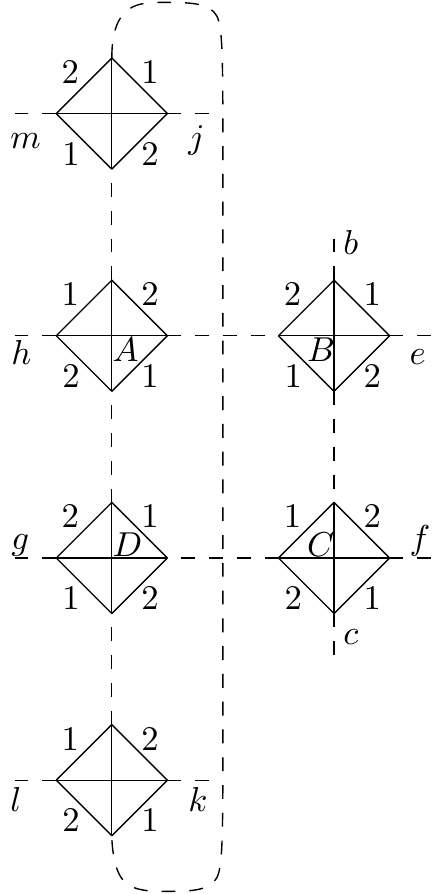} \end{array}
\end{equation}
The face of color 2 which goes along $b$ and $j$ must be of degree 4. The half-edges $b, j$ must therefore be connected to two vertices with an edge of color 2 between them. There are no such vertices available in the subgraph, so a new bubble must be added,
\begin{equation}
\begin{array}{c} \includegraphics[scale=.7]{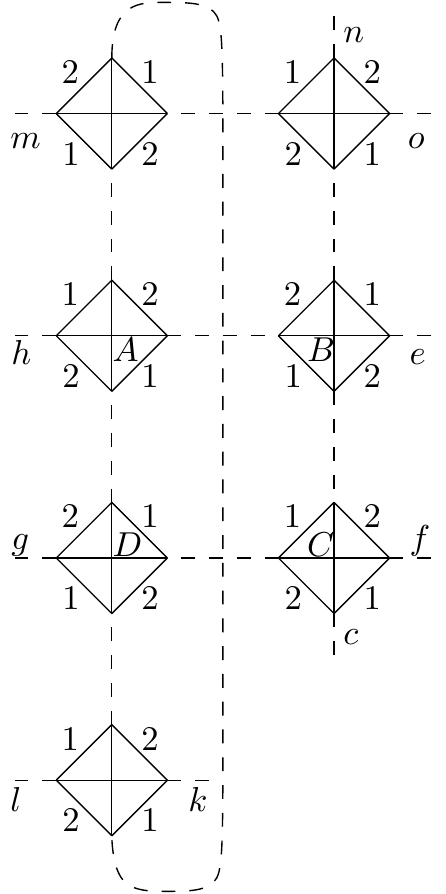} \end{array}
\end{equation}
The face of color 3 which goes along $c$ and $n$ must be of degree 4. The half-edges $c, n$ must therefore be connected to two vertices which are themselves connected by an edge of color 3. There is in our subgraph the edge of of color 3 adjacent to the half-edges $k, l$ available to do so.
\begin{itemize}
\item If $c$ is connected to $k$ and $n$ to $l$, it creates a face of degree 3 (that of color 2 along $c$ and $k$), which is forbidden.
\item If $c$ is connected to $l$ and $n$ to $k$, it creates a face of degree greater than 4 (that of color 2 along $c$), which is forbidden.
\end{itemize}
Therefore another bubble must be added
\begin{equation}
\begin{array}{c} \includegraphics[scale=.7]{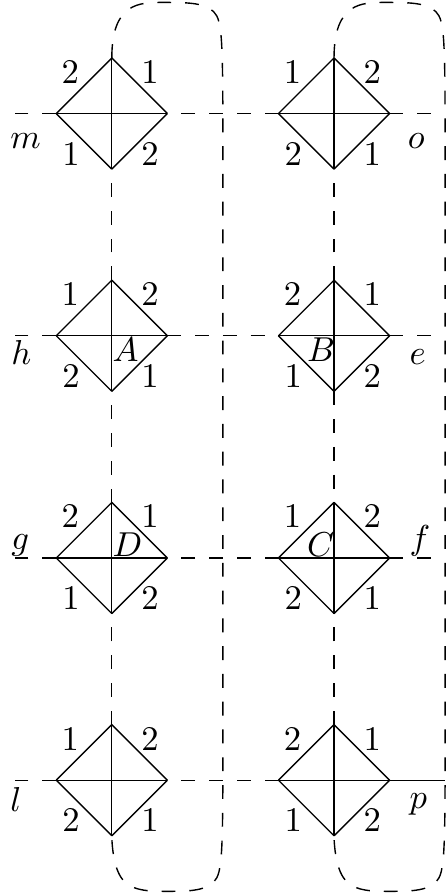} \end{array}
\end{equation}
We now consider the face of color 3 along $h$.
\begin{itemize}
\item If $h$ is connected to $g$, $m$ or $e$, this creates a dipole, which is forbidden.
\item $h$ can be connected to $f$, $l$, $m$, or $p$.
\item $h$ can be to a new bubble.
\end{itemize}

The situations of the second type are all treated similarly and each lead to a single possible graph.
\begin{description}
\item[$h$ to $l$] It then forces $e$ to $p$ so that the face of color 3 along $h$ has degree 4. It then forces $g$ to $m$ so that the face of color 1 along $h$ has degree 4. This in turn forces $f$ to $o$ so that the face of color 3 along $g$ has degree 4 (indeed the bubble labeled $D$ is at distance at least 2 of $B_{\text{exc}}$ so all its incident faces must have degree 4). This fully determines $\cG\in\mathbb{G}_{O(N)^3}$,
\begin{equation}
\begin{array}{c} \includegraphics[scale=.7]{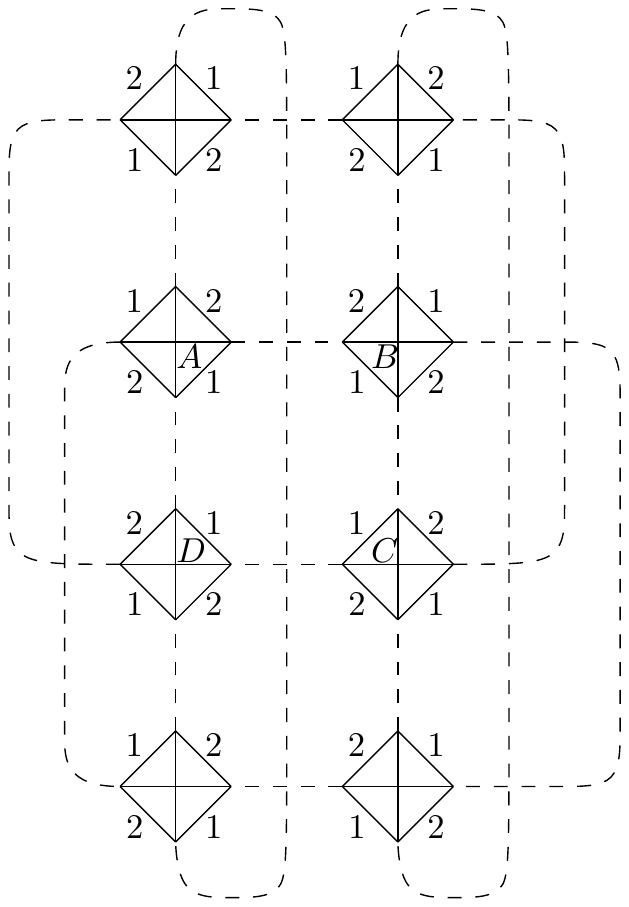} \end{array}
\end{equation}
\item[$h$ to $f$] then forces $g$ to $e$, then $p$ to $m$, then $o$ to $l$.
\item[$h$ to $o$] similar to the previous case by symmetry.
\item[$h$ to $p$] then forces $l$ to $e$, then $f$ to $m$, then $o$ to $g$.
\end{description}

We now consider the case where $h$ is connected to a new bubble. To close the face of color 3 along $h$, it is necessary to have two vertices connected by an edge of color 3, and they cannot belong to the newly added bubble or else the face of degree 4 would be self-intersecting. It is therefore necessary to add a yet another bubble. We arrive at
\begin{equation}
\begin{array}{c} \includegraphics[scale=.7]{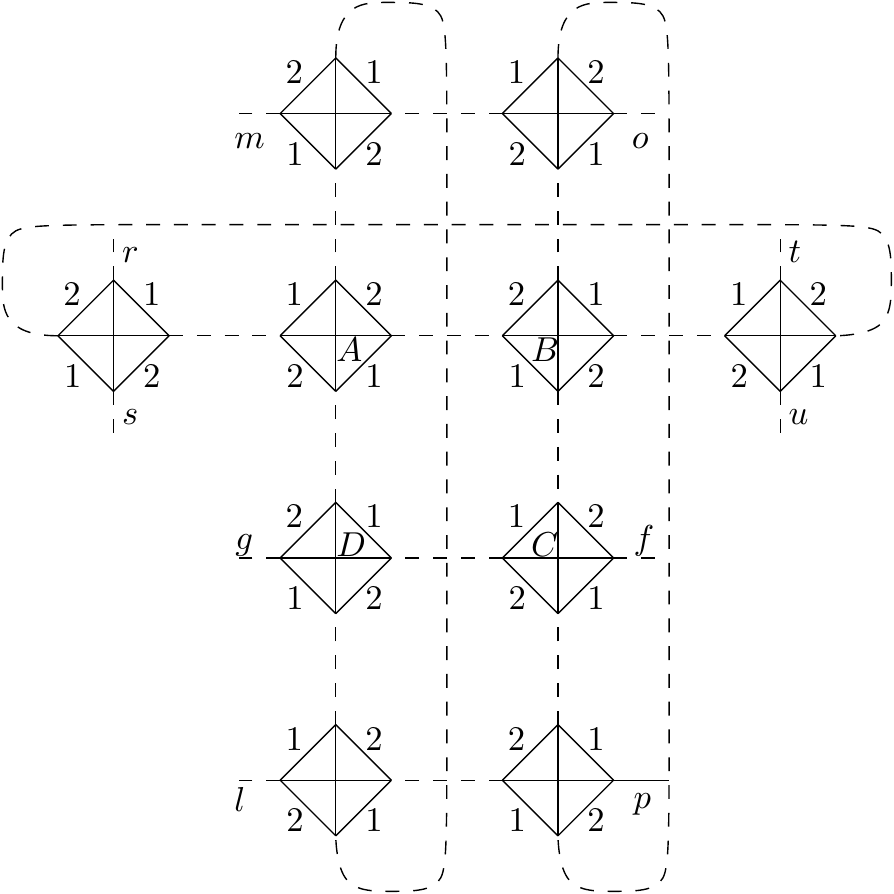} \end{array}
\end{equation}
The face of color 1 along $m$ and $r$ must have degree 4. The half-edges $m, r$ must therefore be connected to the vertices of an edge of 1. Since there is none available, a new bubble must be added. The same holds true for the face of color 2 along $g$ and $s$ (even with the previously added new bubble), so that we get
\begin{equation}
\begin{array}{c} \includegraphics[scale=.7]{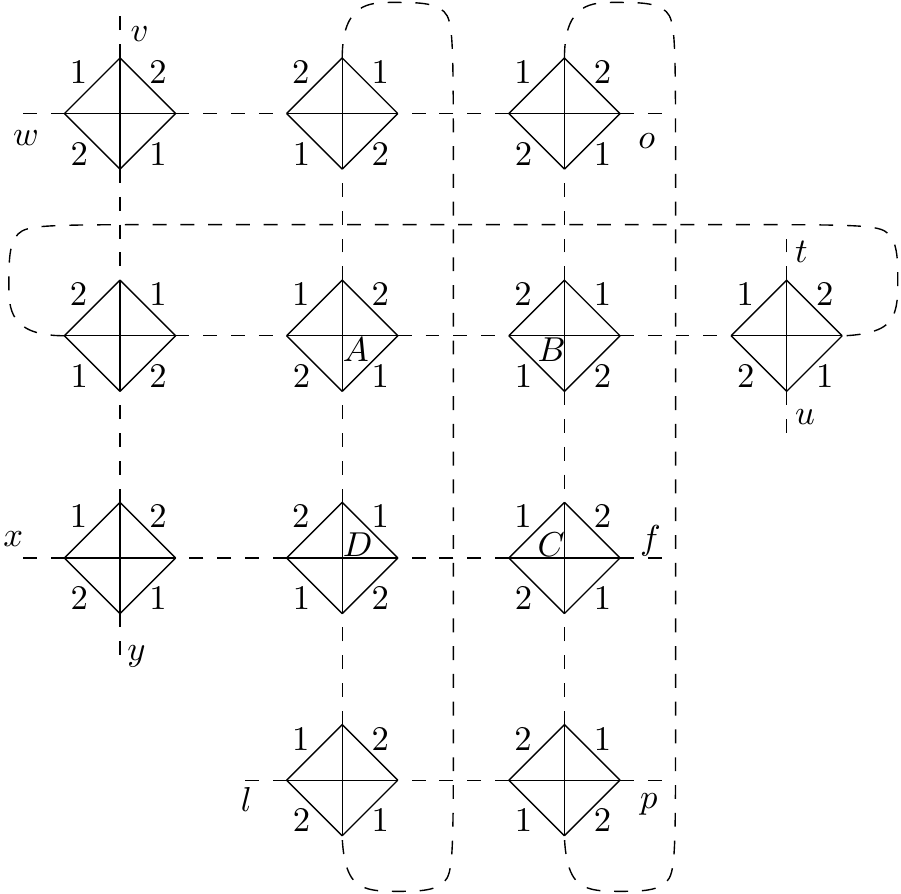} \end{array}
\end{equation}
To close the face of color 1 which goes along $o$ and $t$, one could try to use the edge of color 1 between $v$ and $w$. However, $o$ to $w$ would create a face of color 3 and degree 3, while $o$ to $v$ would create a self-intersecting face of color 3. We therefore need a new bubble. The same argument applies to the face of color 2 along $f$ and $u$ (even with the previously added new bubble, since the latter has no edge of color 2 with both ends available), and we get
\begin{equation}
\begin{array}{c} \includegraphics[scale=.7]{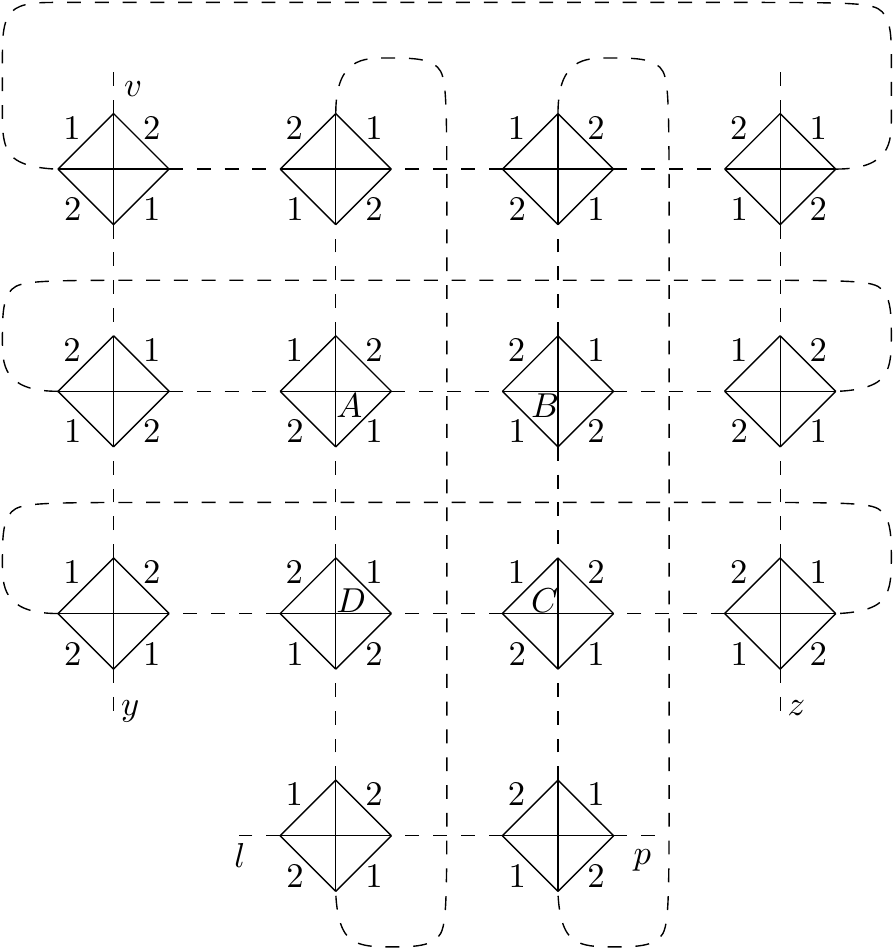} \end{array}
\end{equation}
The face of color 1 going along $l$ and $y$ must have degree 4 (since it is incident to the bubble labeled $D$ which is at distance at least 2 from $B_{\text{exc}}$). The half-edges $l, y$ must therefore be connected to both ends of an edge of color 1, but there is no such edge available in the subgraph. A new bubble must therefore be added. This also allows for closing the face of color 3 along $v$ and $y$ (which must also have degree 4 because the bubble to the left of $A$ is at distance at least 2 of $B_{\text{exc}}$).

The whole argument is then repeated one last time: the face of color 1 along $z$ and $p$ must have degree 4 (since it is incident to the bubble labeled $C$ which is at distance at least 1 from $B_{\text{exc}}$). The half-edges $z, p$ must therefore be connected to both ends of an edge of color 1, but there is no such edge available in the subgraph. A new bubble must therefore be added. It allows for closing the face of color 3 which goes along $z$ (which must have degree 4 because the bubble to the right of $B$ is at distance at least 1 from $B_{\text{exc}}$), and for closing the face of color 3 along $l$ and $p$ (which must have degree 4 because the bubble below $D$ is at distance at least 1 from $B_{\text{exc}}$). This fully determines $\cG$ as
\begin{equation}
\begin{array}{c} \includegraphics[scale=.7]{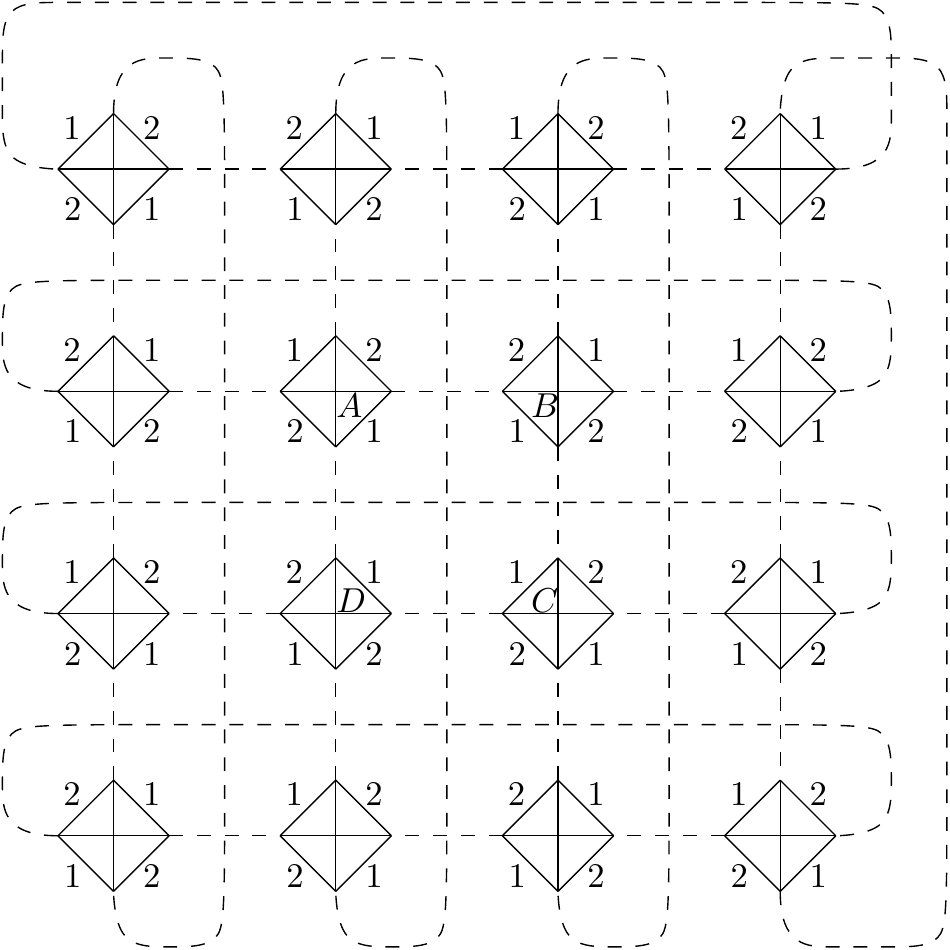} \end{array}
\end{equation}
This exhausts all the possibilities of having a bubble at distance 3 from $B_{\text{exc}}$. They consist in a finite number of graphs. This concludes the proof of Lemma \ref{lemma:sec}.

\section{Identification of the dominant schemes and double scaling}

\subsection{Identification of the relevant singularities} \label{sec:RelevantSingularitiesO(N)3}

From Theorem \ref{th:sch}, we know that there is a finite number of schemes at fixed degree, so the singularities can only come from the generating functions of chains and dipoles and that of melonic 2-point graphs. Therefore, we have a priori three different types of singular points:
\begin{itemize}
\item Singular points of $M(t,\mu)$, which are also singular for $U(t\mu)$ and $B(t,\mu)$.
\item Points such that $U(t,\mu) = 1$, which are singular for any type of chains.
\item Points such that $U(t,\mu) = \frac{1}{3}$ for broken chains only.
\end{itemize}

$M(t,\mu)$ is a generating series whose coefficients $[t^p\mu^q]M$ are the numbers of melonic 2-point graphs with $p$ melons in total and $q$ of type II. Hence its coefficients are all positive. This implies that the function $U_{\mu}: t \mapsto U(t,\mu)$ is an increasing function of $t$. Hence at fixed $\mu$, the point where $U(t,\mu) = \frac{1}{3}$ is always reached for a smaller value of $t$ than $U(t,\mu) = 1$.

Thus, we only have to know whether we first reach a value of $t$ where $U(t,\mu) = \frac{1}{3}$ or $t_c(\mu)$ such that $(t_c(\mu),\mu)$ is a singular point of $M(t,\mu)$. Using equation~\eqref{eq:mel}, we can express the variable $t$ as $ t = \frac{M(t,\mu)-1}{M(t,\mu)^4 + \mu M(t,\mu)^2}$. Therefore the equation $U(t,\mu) = 1/3$ can be written as
\begin{equation}
M(t,\mu) - \frac{4}{3} - \frac{2}{3} \frac{(M(t,\mu)-1)M(t,\mu)^2 \mu}{M(t,\mu)^4 + \mu M(t,\mu)^2} = 0
\end{equation}
Clearing the denominator gives
\begin{equation} \label{CriticalEquation}
-3M(t,\mu)^3 + 4t M(t,\mu)^2 -\mu M(t,\mu) + 2 \mu = 0
\end{equation}
which actually coincides with the equation~\eqref{eq:Mc_poly} determining the critical values of $M(t,\mu)$. Thus, points where $U(t,\mu) = \frac{1}{3}$ are exactly the points which are critical for $M(t,\mu)$ and they are called the \emph{dominant singularities} or \emph{critical curve}. It is plotted in Figure~\ref{fig:plot_tc}.

\begin{figure}
\begin{center}
\includegraphics[scale=0.6]{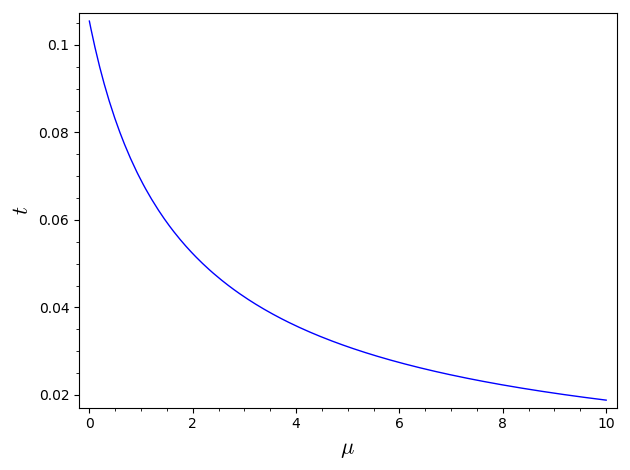}
\caption{Critical points for the generating function $M(t,\mu)$. They also correspond to points where $U(t,\mu) = \frac{1}{3}$ and therefore are critical points for $B(t,\mu)$ as well.}
\label{fig:plot_tc}
\end{center}
\end{figure}

\medskip
Close to the critical curve, the behaviour of $M(t,\mu)$ is given by equation~\eqref{eq:crit_behav}. Therefore $U(t,\mu)$ can be expressed as
\begin{equation}
U(t,\mu) \underset{t \rightarrow t_c(\mu)}{\sim} \frac{1}{3} + \left(1 - \frac{4}{3}t_c(\mu)\mu M_c(\mu)\right) K(\mu)\sqrt{1-\frac{t}{t_c(\mu)}} - \frac{2}{3} t_c(\mu) \mu K(\mu)^2 \left(1-\frac{t}{t_c(\mu)}\right)
\end{equation}

It follows that for $\mu \geq 0$, the behaviour of $B$ near the critical curve is given by
\begin{align}
B(t,\mu) &\underset{t\rightarrow t_c(\mu)}{\sim} \frac{1}{\left(1 - \frac{4}{3}t_c(\mu)\mu M_c(\mu)\right) K(\mu)\sqrt{1-\frac{t}{t_c(\mu)}}}
\label{eq:B_crit}
\end{align}

\subsection{Identification of the dominant schemes}


To perform the double-scaling limit, we identify the schemes at fixed $\omega$ which are the most singular at criticality, and call them \emph{dominant schemes}. From the above analysis, they are the schemes which maximize the number of broken chains. The analysis of \cite{TaFu} can be applied almost verbatim. It results in Theorem \ref{thm:DominantSchemesO(N)3} below.

Let us recall that a tree is a graph with no cycles. Vertices of valency 1 are called leaves and the others are called internal nodes. A rooted tree is a tree with a marked leaf. A binary tree is a tree whose internal nodes all have valency 3. A tree is said to be plane if it is embedded in the plane. 

\begin{theorem} \label{thm:DominantSchemesO(N)3}
The dominant schemes of degree $\omega>0$ are given bijectively by rooted plane binary trees with $4\omega-1$ edges, with the following correspondence
\begin{itemize}
\item The root of the tree corresponds to the two external legs of the 2-point function.
\item Edges of the tree correspond to broken chains.
\item The leaves are tadpoles: $\begin{array}{c} \includegraphics[scale=.4]{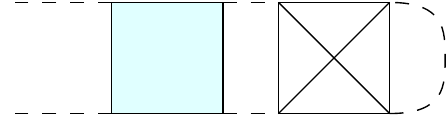}\end{array}$
\item There are two types of internal nodes,
\begin{equation} \label{InnerNodes}
\begin{array}{c} \includegraphics[scale=.6]{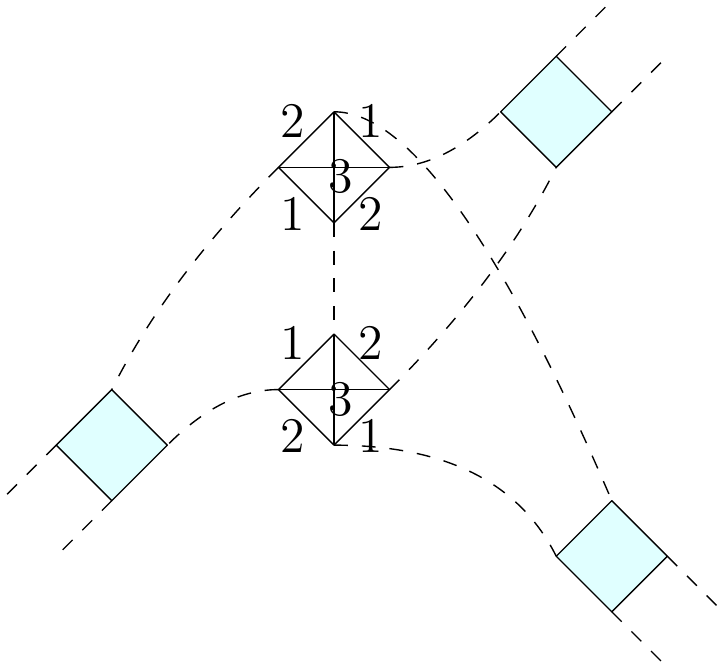}\end{array}, \qquad \begin{array}{c} \includegraphics[scale=.6]{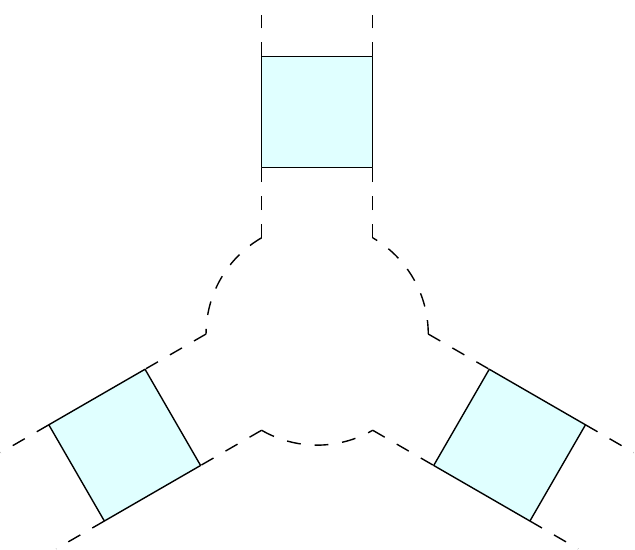}\end{array}
\end{equation}
\end{itemize}
\end{theorem}

The degree is thus entirely ``contained'' in tadpoles at the end of the broken chains.

\begin{proof}
We use the notations of Lemma~\ref{thm:SkeletonGraph}. From this lemma, it is clear that if $\cS$ is dominant, then its skeleton graph is a tree, i.e. $\cI(\cS) = \cT$ (or $q=0$), with $N(\cT)=4\omega(\cT)-1$ edges, all corresponding to  broken chains in $\cS$. Moreover
\begin{itemize}
\item all its internal nodes have valency 3 and correspond to components $\bar{\cG}^{(r)}$ of degree 0,
\item all leaves correspond to components $\bar{\cG}^{(r)}$ of degree 1/2,
\item the component $\cG^{(0)}$, which has the two external legs, gives rise to a root for $\cT$, and has degree 0.
\end{itemize}
We now have to identify which graphs can appear as nodes and leaves of the tree.
\begin{itemize}
\item Leaves correspond to graphs of degree $1/2$, which have been identified in~\cite{TaCa}. They are tadpoles, i.e. graphs with one tetrahedral interaction and two edges of color 0. There are three different tadpoles, depending on the color of the two faces of length $1$. One cuts an edge of color 0 and connect them to one side of a chain in $\cS$.
\item An internal node of $\cT$ corresponds to a graph of degree 0, with three edges of color 0 cut in order to connect it to three chains. The resulting 6-point function must have no melons and no dipoles. One can check that it must have either $0$ interactions, i.e. be a single propagator closed on itself, giving rise to the second type of vertices in \eqref{InnerNodes}, or 2 interactions, which is the elementary type I melon, giving rise to the first type of vertices in \eqref{InnerNodes}.
\end{itemize}
Clearly, the dominant schemes are fully encoded by their skeleton graphs which are rooted binary trees, except for the order of the chains meeting an internal nodes. We thus obtain a bijection between the dominant schemes and rooted binary trees with an order of the edges incident at every vertex, i.e. plane trees. This concludes the proof of the theorem.
\end{proof}

\subsection{Generating function for the dominant schemes}

A rooted binary tree with $N$ edges has $\frac{N-1}{2}$ inner nodes and $\frac{N+1}{2}$ leaves (not counting the root). Here a leaf carries a weight $3t^{1/2}$ (for the three ways to form a tadpole). An inner node receives a weight $1+6t$, the 1 being due to the second type of nodes in \eqref{InnerNodes} and the $6t$ to the first type.

A dominant scheme corresponds to a rooted, plane binary tree $\mathcal{T}$ with $4\omega-1$ edges and thus its generating function is
\begin{equation}
G_{\mathcal{T}}^{\omega}(t,\mu) = (3t^{\frac{1}{2}})^{2\omega} (1+6t)^{2\omega-1} B(t, \mu)^{4\omega-1} = (3t^{\frac{1}{2}})^{2\omega} (1+6t)^{2\omega-1} \frac{6^{4\omega-1}U^{8\omega-2}}{\left((1-U)(1-3U)\right)^{4\omega-1}}
\end{equation}
This function only depends on $\omega$ and not on the shape of $\mathcal{T}$. We can therefore easily sum over all trees and also add all melonic insertions at the root,
\begin{equation}
G_{\text{dom}}^{\omega}(t,\mu) = M(t, \mu) \sum_{\substack{\mathcal{T}\\ \text{$2\omega$ leaves}}} G_{\mathcal{T}}^{\omega}(t,\mu) = \Cat_{2\omega-1} G_{\mathcal{T}}^{\omega}(t,\mu) 
\label{eq:fct_dom_scheme}
\end{equation}
where $\Cat_{2\omega-1} = \frac{1}{2\omega}\binom{4\omega-2}{2\omega-1}$ is the number of rooted, plane, binary trees with $2\omega$ leaves.

\subsection{Double scaling limit of the quartic $O(N)^3$ model}

Using Equation~\eqref{eq:crit_behav}, the generating function of dominant schemes behaves near singular points as
\begin{equation}
G_{dom}^\omega(t,\mu) \underset{t\rightarrow t_c}{\sim} M_c(\mu)\Cat_{2\omega-1}9^{\omega}t_c^{\omega}\left(1+6t_c\right)^{2\omega-1}\left(\frac{1}{\left(1 - \frac{4}{3}t_c(\mu)\mu M_c(\mu)\right) K(\mu)\sqrt{1-\frac{t}{t_c(\mu)}}}\right)^{4\omega-1}
\end{equation}
Since in the large $N$ expansion a graph $\cG$ of degree $\omega$ scales as $N^{3-\omega}$ we define the following double scaling parameter
\begin{equation}
\kappa(\mu)^{-1} = N^{\frac{1}{2}}\frac{1}{3}\frac{1}{t_c(\mu)^\frac{1}{2}\left(1+6t_c(\mu)\right)}\left( \left(1 - \frac{4}{3}t_c(\mu)\mu M_c(\mu)\right) K(\mu) \right)^2\left(1-\frac{t}{t_c(\mu)}\right)
\label{eq:kappa}
\end{equation}
Using Equation~\eqref{eq:kappa} we get:
\begin{equation}
\left[ (1+6t_c(\mu)) \frac{1}{\left(1 - \frac{4}{3}t_c(\mu)\mu M_c(\mu)\right) K(\mu)\sqrt{1-\frac{t}{t_c(\mu)}}}\right]^{-1} = \kappa(\mu)^{-\frac{1}{2}}N^{-\frac{1}{4}}\frac{\sqrt{3}t_c(\mu)^\frac{1}{4}}{\left(1+6t_c(\mu)\right)^\frac{1}{2}}
\end{equation}

Therefore in the double scaling limit, the dominant graphs of degree $\omega > 0$ contribute as:
\begin{equation}
G_{dom}^\omega(\mu) = M_c(\mu)\frac{N^{\frac{11}{12}}}{\kappa(\mu)^\frac{1}{2}}\sqrt{3}\frac{t_c(\mu)^\frac{1}{4}}{\left(1+6t_c(\mu)\right)^\frac{1}{2}}\Cat_{2\omega-1}\kappa(\mu)^{2\omega}
\end{equation}
where one has to add the contribution of the graphs of degree $0$ i.e. the melons, which contribute simply as $M_c(t,\mu)$.

\vspace{15pt}
Hence summing over contribution of all degree, the total contribution to $G_2^{DS}$ is:
\begin{align}
G_{2}^{DS}(\mu) &= \sum\limits_{\omega\in\mathbb{N}/2} G_{dom}^\omega(\mu) \nonumber \\
&= M_c(\mu) + M_c(\mu)\frac{N^{\frac{11}{12}}}{\kappa(\mu)^\frac{1}{2}}\sqrt{3}\frac{t_c(\mu)^\frac{1}{4}}{\left(1+6t_c(\mu)\right)^\frac{1}{2}} \sum\limits_{n \in \frac{\mathbb{N}}{2} > 0 } \Cat_{2\omega-1}\kappa(\mu)^{2\omega} \nonumber \\
&= M_c(\mu) + M_c(\mu)\kappa(\mu)N^{\frac{11}{12}}\sqrt{3}\frac{t_c(\mu)^\frac{1}{4}}{\left(1+6t_c(\mu)\right)^\frac{1}{2}} \sum\limits_{n \in \mathbb{N}} \Cat_n \kappa(\mu)^{n} \nonumber \\
&= M_c(\mu) \left(1 + N^{\frac{11}{12}}\sqrt{3}\frac{t_c(\mu)^\frac{1}{4}}{\left(1+6t_c(\mu)\right)^\frac{1}{2}} \frac{1-\sqrt{1-4\kappa(\mu)}}{2\kappa(\mu)^\frac{1}{2}}\right).
\end{align}

Let us give now an analysis of this final result. Note that the sum converges for $\kappa(\mu) \leq \frac{1}{4}$, thus showing, as announced above, that the tensor double scaling series of the $O(N)^3$-invariant tensor model is convergent. This last identity is the identity announced in Theorem \ref{resultatfinal} stated in the Introduction. This is a different type of result when compared to the matrix case, where the double scaling limit series is divergent. 

Moreover, the parameter $\kappa(\mu)$ allows to define a double scaling limit such that graphs of all orders in the $\frac{1}{N}$ expansion contribute. This conclusion comes from an analysis  of the singularities arising from (broken) chains: higher order graphs in the $\frac{1}{N}$ expansion can have more (broken) chains. Thus, we can tune how to approach the critical point while sending $N$ to infinity such that the singularities coming from the chains compensate the loss of scaling in $N$. 
It is the double scaling parameter $\kappa(\mu)$ which encodes this balance between the large $N$ limit and the criticality of chains. In particular, when $\kappa(\mu) \rightarrow 0$, we obtain the usual melonic large $N$ limit.
Another interesting value is $\kappa(\mu) = \frac{1}{4}$,
when analyticity is lost. 

For the quartic $O(N)^3-$invariant tensor model, $\kappa(\mu)$ explicitly depends on the ratio $\mu$ of the two coupling constants through $M_c(\mu)$, $t_c(\mu)$ and $K(\mu)$, due to the presence of the pillow interaction. In particular, when $\mu = 0$, we obtain a similar expression 
to the one that has been obtained for other tensor models with tetrahedral interaction~\cite{GuTaYo,BeCa}.

\bigskip

\paragraph*{Acknowledgments.} The authors have been partially supported by the ANR-20-CE48-0018 "3DMaps" grant.

\begin{figure}
\includegraphics[scale=0.5]{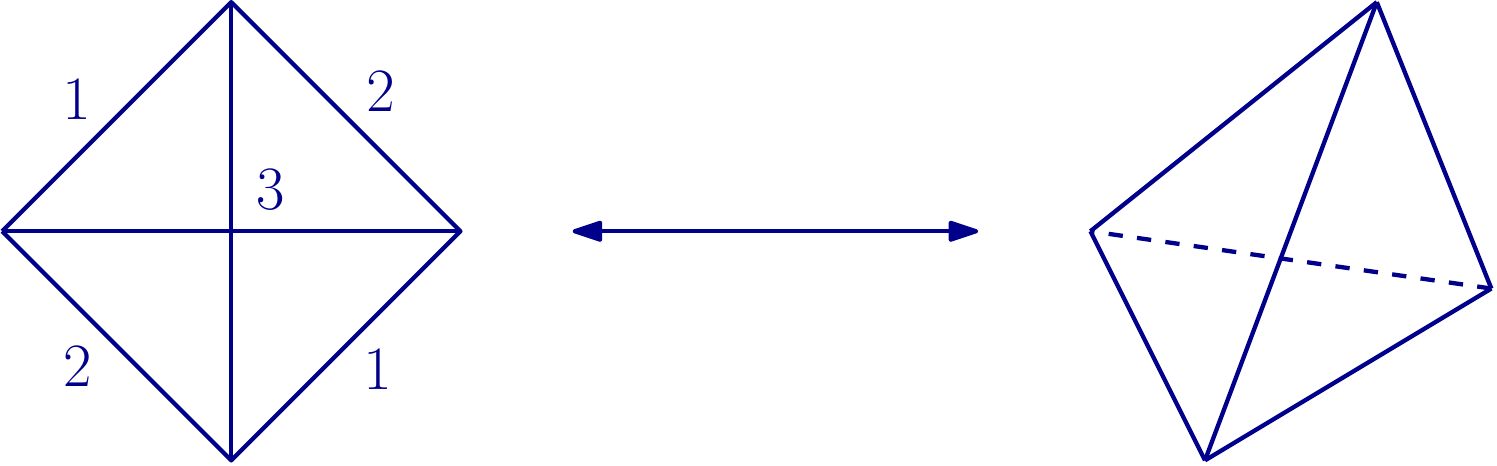}
\end{figure}
 \appendix
 
\bibliography{DS_biblio}

\end{document}